\documentclass[12pt,article]{amsart}
\usepackage{mathrsfs}
\usepackage{amssymb}
\usepackage{amsfonts}
\usepackage{amsbsy}
\usepackage{latexsym}
\usepackage{amssymb,latexsym,amsmath,amsthm}
\usepackage{framed}
\usepackage[colorlinks,linkcolor=blue]{hyperref}
\usepackage{graphicx}
\usepackage{xcolor}
\usepackage{epstopdf}
\usepackage{bm}
\theoremstyle{plain}

 \theoremstyle{remark} 

\newtheorem {theo} {\bf Theorem} [section]
\newtheorem {prop} [theo] {\bf Proposition}

\newtheorem {lem} [theo] {\bf Lemma}

\newtheorem{exam} {\bf Example}[section]

\newtheorem{appr} {\bf Approach}[section]

\newtheorem{rem}{\bf Remark}[section]

\numberwithin{equation}{section}
\begin{document}
\title[FROG-PR of analytic signals]{FROG-measurement  based  phase  retrieval  for analytic  signals}
\author{Youfa Li}
\address{College of Mathematics and Information Science\\
Guangxi University,  Nanning, China,530004 }
\email{youfalee@hotmail.com}
\author{Yaoshuai Ma}
\address{College of Mathematics and Information Science\\
Guangxi University,  Nanning, China,530004 }
\email{MYShuai@163.com}
\author{Deguang Han}
\address{
University of Central Florida,  Orlando, FL 32816}
\email{Deguang.Han@ucf.edu}
\thanks{Youfa Li is partially supported by Natural Science Foundation of China (Nos: 61961003, 61561006, 11501132),  Natural Science Foundation of Guangxi (Nos: 2018JJA110110, 2016GXNSFAA380049) and  the talent project of  Education Department of Guangxi Government  for Young-Middle-Aged backbone teachers.
Deguang Han  is partially supported by the NSF grant  DMS-1712602.
}

\date{\today}

\begin{abstract}
While frequency-resolved optical gating (FROG) is widely used in characterizing the
ultrafast pulse in optics,  analytic signals are often considered in time-frequency analysis and signal processing, especially when  extracting instantaneous features of events.
In this paper we examine the
phase retrieval (PR) problem of analytic signals in $\Bbb{C}^N$ by their  FROG measurements. After establishing the ambiguity of
the FROG-PR of analytic signals, we found that the FROG-PR of analytic signals of even
lengths is different from that of analytic signals of odd lengths, and it is also different from the case of $B$-bandlimited signals with $B \leq N/2$. The existing
approach to bandlimited signals can be applied to analytic signals of odd lengths, but it does not apply to the even length case. With the help of two relaxed FROG-PR problems and a translation technique,
we  develop an approach  to FROG-PR for  the analytic signals of even lengths, and prove that  in this case the  generic analytic signals can be  uniquely (up to the ambiguity) determined  by their  $(3N/2+1)$ FROG measurements.
\end{abstract}
\maketitle

\section{Introduction}\label{jibenzhunbei}

The classical phase retrieval (PR)  is a nonlinear   problem  that seeks to
reconstruct   a signal   $\textbf{z}=(z_{0}, \ldots, z_{N-1})$
$\in \mathbb{C}^{N}$ (\emph{up to  the ambiguities})  from   the intensities  of  its   Fourier measurements
(c.f.\cite{Ba4,Fienup1,Fienup2,Fienup3,Gerchberg,reasonforPR})
\begin{align}\notag \begin{array}{lll}  b_{k}:=\big|\sum_{n=0}^{N-1}z_{n}e^{-2\pi \textbf{i}kn/N}\big|, \ k\in \Gamma. \end{array}\end{align}
Here the ambiguity (c.f. \cite{Plonka}) means that,
there exist other  signals in $\mathbb{C}^{N}$ such that they have the same
intensity measurements  as $\textbf{z}$.

 PR   has   been
widely applied   to   engineering   problems such as the  coherent diffraction imaging (\cite{reasonforPR}) and    quantum tomography (\cite{Heinosaarri}). To be adapted to more  applications,  it    has been investigated by non Fourier   measurements
such as frame intensity measurements (e.g. \cite{HanDe1,Ba2,Ba3,Saving,Flinth,Bodmann,Science61,BAS100,Xu1,HanDe3,Han,Xu2,Xu3}) and phaseless sampling (e.g. \cite{ QiyuPR,QiyuPR1,LSUN}). In this paper we  examine the  PR problem related to the FROG (frequency-resolved optical gating)  measurements (c.f. \cite{Eldar,21,22}) for a special type of signals.


Given a pair  $[N, L] \in \mathbb{N}^{2}$ such that $L\leq N$,
denote $r:=\lceil N/L\rceil$.
For $\textbf{z}=(z_{0}, \ldots, z_{N-1})\in \mathbb{C}^{N}$, $k=0,1,\cdots,N-1$ and  $m=0,1,\ldots, r-1$,
define $y_{k,m}:=z_{k}z_{k+mL}$.  Then
the {\bf  $[N, L]$-FROG measurement}  of $\textbf{z}$ at $(k, m)$ is  defined as
\begin{equation}\label{frogmeasuement}\begin{array}{lllll}
|\hat{y}_{k,m}|^{2}&=\big|\sum_{n=0}^{N-1}y_{n,m}e^{-\textbf{i}2\pi kn/N}\big|^{2}\\
&=\big|\sum_{n=0}^{N-1}z_{n}z_{n+mL}e^{-\textbf{i}2\pi kn/N}\big|^{2}.
\end{array}\end{equation}
 It follows from \cite[(3.1)]{Eldar} that
 \begin{equation}\label{FROGFREQ}
\hat{y}_{k,m} =\frac{1}{N}\sum_{l=0}^{N-1}\hat{z}_{l}\hat{z}_{k-l}w^{lm},
\end{equation}
where $w=e^{\textbf{i}2\pi/r}$.  Now the FROG-PR asks to determine $\textbf{z}$ by its FROG measurements
$\{|\hat{y}_{k,m}|^{2}\}$, up to the ambiguity which might be different depending on the type of signals we are working with.
This problem was recently investigated by T. Bendory, D. Edidin  and Y.C. Eldar  \cite{Eldar} for $B$-bandlimited signals in $\mathbb{C}^{N}$. When  $B\leq N/2$,  the ambiguity arises from arbitrary rotation,   arbitrary translation and reflection.

\begin{theo} \cite{Eldar}  Assume that $N/L \geq 4$ and $B\le N/2$. Then generic  $B$-bandlimited signals  in $\Bbb{C}^{N}$ are uniquely (up to their ambiguity)  determined  from $3B$ FROG measurements.
\end{theo}

 One of the key steps in the proof  in  \cite[section 3.2]{Eldar}
for the FROG-PR is to assign   $0$
or $\pi$ to   both  the two phases   $\arg(\hat{z}_{0})$ and $\arg(\hat{z}_{1})$. Such an assignment  holds due to the above ambiguity.
In this paper we are interested in investigating analytic signals whose ambiguity varies depending on the lengths of the signals. 


Let $\textbf{x}\in \mathbb{R}^N$  be a real-valued signal with {\bf
discrete Fourier transform  (DFT)} $\widehat{\textbf{x}}=(\hat{x}_{0}, \ldots, \hat{x}_{N-1})$. By L. Marple \cite{analtyic},  the analytic signal $A\textbf{x}=(Ax_{0},
\ldots, $
$Ax_{N-1})$ corresponding to $\textbf{x}$ is defined through its DFT $\widehat{A\textbf{x}}=(\hat{Ax}_{0},
\ldots, \hat{Ax}_{N-1})$, where   for  even length  $N$,
\begin{align}\label{up1}\hat{Ax}_{k}=
\left\{\begin{array}{lll}
\hat{x}_{0},&k=0, \\
2\hat{x}_{k},&1\le k \le N/2-1,\\
\hat{x}_{N/2},&k=N/2,\\
0,&N/2+1 \le k \le N-1,
\end{array}\right.
\end{align}
and for odd length  $N$,
\begin{align} \hat{Ax}_{k}=
\left\{\begin{array}{lll}
\hat{x}_{0},&k=0, \\
2 \hat{x}_{k},&1\le k \le (N-1)/2,\\
0,&(N+1)/2 \le k \le N-1.
\end{array}\right.
\end{align}
By   \cite{analtyic}, the real part  $\Re(A\textbf{x})=\textbf{x}$ and the imaginary part  $\Im(A\textbf{x})$
is the discrete Hilbert transform of $\textbf{x}$. Moreover, the inner product   $
   \langle\Re(A\textbf{x}), \Im(A\textbf{x})\rangle=0$.


Analytic signals form an important class of signals that  have been   widely used in time-frequency analysis and signal processing, especially in  extracting instantaneous features
(e.g. \cite{analyticsignal3,analyticsigna1,analyticsigna2,analyticsignal4}).  In order to examine the phase retrieval (up to ambiguity) problem for analytic signals from their FROG measurements, we need first to identify their ambiguity of FROG-PR measurements, which will be described in  Proposition \ref{even} and Proposition \ref{oddlest}. A key point is that the FROG-PR  ambiguity of analytic signals of even lengths is  different from that of analytic signals of odd lengths, and consequently they are different from that of $B$-bandlimited signals in $\mathbb{C}^{N}$ \cite{Eldar} with $B \leq N/2$. In the next section we will explain while the approach from \cite{Eldar} can be applied to the signals of odd lengths,  even length analytic signals are quite different from the odd length ones and  it requires a different approach for such signals. So we only need to focus on analytic signals of even length\textbf{s}. The following is  the main result (see Theorem \ref{twobridges}) of this paper.

\begin{theo} Assume that $N$ is even, $L$ is odd and $r=\lceil N/L\rceil\geq 5$. Then generic analytic signals in $\mathbb{C}^{N}$  can be determined  uniquely up to the ambiguity in Proposition \ref{even} by their  $(3N/2+1)$ $[N, L]$-FROG measurements.
\end{theo}

%


Note that for \textbf{generic}  analytic signals in $\Bbb{C}^N$ of even lengths, the  bandlimit is $B = N/2 +1$. The above theorem tells us that 
such an analytic signal can be uniquely (up to the ambiguity) determined  by $3B - 2$ number of FROG measurements. Moreover,  the procedures for
such a determination will be provided by Approach \ref{APPROACH3}.

\vspace{3mm}

\noindent{\bf Outline of the paper.} In Section 2 we first present a characterization of analytic signals, and then establish the ambiguity of FROG-PR based on the parity of the lengths of the analytic signals. A rationale is provided at the end of this section to explain why we need a different approach for the even length case. Section 3 is devoted to presenting the step-by-step approach that eventually led to the main theorem. The technical proofs for most of the theoretical preparations that led to the main result (Theorem \ref{twobridges} ) are presented in Section 4. At the end of the paper, we point out that our approach fails for the case when $L$ is even, and further investigation is needed to address this case.

%

\section{The ambiguity of FROG-PR  of  analytic signals}

We first fix a few standard notations.
A complex   number  $0\neq z=\Re(z)+\textbf{i}\Im(z)\in \mathbb{C}$  can be  denoted by   $|z|e^{\textbf{i}\arg(z)}$,  where    $\textbf{i}$, $|z|$ and  $\arg(z)$ are   the imaginary unit,   modulus and     phase, respectively.
The  conjugation  of a complex vector   $\textbf{z}=(z_{0}, \ldots,$
$ z_{N-1})\in \mathbb{C}^{N}$
is denoted by
$\bar{\textbf{z}}:=(\bar{z}_{0}, \ldots, \bar{z}_{N-1}),$
where $\bar{z}_{k}$ is the complex conjugation    of $z_{k}.$
Throughout this paper, the signal $\textbf{z}$ is   $N$-periodic, namely,
$z_{\ell}=z_{\ell+N}$ for all $\ell \in \mathbb{Z}.$
From now on, its DFT is defined  by
$\widehat{\textbf{z}}:=(\hat{z}_{0}, \ldots, \hat{z}_{N-1})$ with
$\hat{z}_{k}= \sum_{n=0}^{N-1}z_{n}e^{-2\pi \textbf{i}kn/N}.$
And through the inverse discrete Fourier transform (IDFT),
$\textbf{z}$ can be reconstructed by $z_{k}= \frac{1}{N}\sum_{n=0}^{N-1}\hat{z}_{n}e^{2\pi \textbf{i}kn/N}.$
For any $\textbf{z}\in \mathbb{C}^{N}$,   its $\theta$-rotation  $\textbf{z}^{\theta}_{ro}:=e^{\textbf{i}\theta}\textbf{z}$,
reflection  $\textbf{z}_{ref}:=(\overline{z_{0}}, \overline{z_{-1}},  \ldots, \overline{z_{-(N-1)}})$   and
for any  $\gamma\in \mathbb{R}$, its $\gamma$-translation  $\textbf{z}^{\gamma}_{tr}$
 is defined through   $\widehat{\textbf{z}^{\gamma}_{tr}}=(\hat{z}_{0}, \hat{z}_{1}e^{\textbf{i}2\pi\gamma/N},
 \ldots, $
 $\hat{z}_{N-1}e^{\textbf{i}2\pi(N-1)\gamma/N})$. We say that  $\textbf{z}$
 is $B$-bandlimited if $\widehat{\textbf{z}}$ contains $N-B$ consecutive zeros (c.f. \cite{Eldar}).
Denote by $\#\Lambda$
the cardinality  of a set  $\Lambda$, and  by $\lceil x\rceil$
the smallest number that is not  smaller than $x\in \mathbb{R}$.

The following characterizes analytic signals,  and its proof will be  presented in section \ref{appendix1}.

\begin{prop}\label{characterizationeven}
Suppose that $\textbf{z} \in \mathbb{C}^{N}$.  Denote  the Cartesian product of sets by $\times$.
Then  $\textbf{z}$ is   analytic   if and only if one of the following two items   holds:\\
(i) for even length  $N$, $\widehat{\textbf{z}}\in \mathbb{R} \times \mathbb{C}^{N/2-1}\times \mathbb{R}\times \overbrace{\{0\}\times \ldots\times \{0\}}^{(N/2-1) \hbox{copies}}$;\\
 (ii) for odd length  $N$, $\widehat{\textbf{z}}\in \mathbb{R} \times \mathbb{C}^{(N-1)/2}\times \overbrace{\{0\}\times \ldots\times \{0\}}^{[(N-1)/2] \hbox{copies}}$.
\end{prop}




As a consequence of Proposition \ref{characterizationeven}, we get
 \begin{prop}\label{remm}
 Suppose that $\textbf{z} \in \mathbb{C}^{N}$ is  analytic.
  Then the following two items  hold.\\
  (i)  If $N$ is even and $\widehat{\textbf{z}}=(\hat{z}_{0}, \ldots, \hat{z}_{N/2}, 0, \ldots, 0)$
satisfies $\hat{z}_{N/2}\neq0$, then  the  $\gamma$-translation $\textbf{z}^{\gamma}_{tr}$ is still analytic if and only if $\gamma\in \mathbb{Z}.$\\
 (ii)
If  $\hat{z}_{0}\neq0$,   then  its $\theta$-rotation  $e^{\textbf{i}\theta}\textbf{z}$
is still  analytic if and only if $\theta=k\pi$ for $k\in \mathbb{Z}.$

 \end{prop}

Proposition \ref{remm} (i) implies that in the even length case,   the  non-integer translation does not
inherit the analytic property. Moreover, the following example shows  that
 if $\gamma\not\in \mathbb{Z}$, then the $\gamma$-translation $\textbf{z}^{\gamma}_{tr}$
does not necessarily  have the same FROG measurements as $\textbf{z}$.

\begin{exam}\label{kewangyouhua}
Let
$\textbf{x}=(0.3252, -0.7549, 1.3703,
-1.7115)\in \mathbb{R}^{4}$.
Then its analytic signal is given by  $\textbf{z}=A\textbf{x}=(0.3252 - 0.4783\textbf{i},  -0.7549 - 0.5226\textbf{i},   1.3703 + 0.4783\textbf{i},  -1.7115 + 0.5226\textbf{i})$.  By DFT, we have  $\widehat{\textbf{z}}=(-0.7710, -2.0902 - 1.9132\textbf{i},   4.1619,$
$   0)$.
Choose $\gamma=\frac{2}{\pi}$. Then we have $\widehat{\textbf{z}^{\gamma}_{tr}}=(-0.7710,   0.4805-2.7925\textbf{i},  -1.7320 + 3.7844\textbf{i},$
$   0)$ and  $\textbf{z}^{\gamma}_{tr}=(-0.5056 + 0.2480\textbf{i},   0.9384 - 0.8260\textbf{i},  -0.7459 + 1.6442\textbf{i},  -0.4579 - 1.0662\textbf{i})$. Take  $L=1$ for example. By direct  calculation,
the FROG measurement of $\textbf{z}$ at $(0, 0)$ is  $|\hat{y}_{0,0}|^{2}=20.0614$
while that of $\textbf{z}^{\gamma}_{tr}$ is $|\hat{y^{\gamma}}_{0,0}|^{2}=17.9335$.
Thus the $\textbf{z}$ and $\textbf{z}^{\gamma}_{tr}$ do not  have the same FROG measurements.
\end{exam}

The following two propositions establish the ambiguity of FROG-PR of analytic signals.

\begin{prop}\label{even}
Suppose that   $\textbf{z} \in \mathbb{C}^{N}$ is   analytic   such that  $N$ is even. Then its  $\pi$-rotation $-\textbf{z}$,
integer-translation $\textbf{z}^{l}_{tr}$ with $l\in \mathbb{Z}$, and reflection  $\textbf{z}_{ref}$
are all analytic. Moreover, they have the same FROG measurements as $\textbf{z}$.
\end{prop}
\begin{proof}
By Proposition  \ref{remm}, both  $-\textbf{z}$ and  $\textbf{z}^{l}_{tr}$
are analytic. By Proposition \ref{characterizationeven} (i), $\textbf{z}_{ref}$
is also analytic.
Moreover, by \cite[Proposition 2.2]{Eldar} they have the same FROG measurements as $\textbf{z}$.
\end{proof}



\begin{prop}\label{oddlest}
Suppose that   $\textbf{z} \in \mathbb{C}^{N}$ is   analytic   such that  $N$ is odd. Then its   $\pi$-rotation  $-\textbf{z}$,
any translation  $\textbf{z}^{\gamma}_{tr}$ with $\gamma\in \mathbb{R}$, and reflection  $\textbf{z}_{ref}$
are all analytic. Moreover, they have the same FROG measurements as $\textbf{z}$.
\end{prop}


By Propositions \ref{even} and \ref{oddlest},  the FROG-PR  ambiguity of analytic signals of odd lengths is essentially
different from that of analytic signals of even lengths, and  the ambiguity of analytic signals for the even length case
is also different from ambiguity of signals investigated in \cite[Proposition 2.2]{Eldar}.   
Therefore it is expected that different approaches might be needed for analytic signals of different lengths. Next we  explain that while  the FROG-PR of analytic signals of odd lengths
can be achieved through the similar procedures as those in \cite[section 3.2]{Eldar}, such  procedures do not hold for analytic signals of even lengths.
\vspace{3mm}

{\bf Case I: $N$ is odd.} Suppose that $\textbf{z} \in \mathbb{C}^{N}$ is analytic  and  $N$ is odd.
Through the direct calculation, the equation system \eqref{FROGFREQ} enjoys the following "pyramid" structure
w.r.t   variables $\hat{z}_{0}, \ldots, \hat{z}_{(N-1)/2}$:
\begin{center}
\begin{gather}
\hat{z}_{0}^{2},0,0,\cdots,0,0,\cdots,0  \nonumber \\
\hat{z}_{0}\hat{z}_{1},\hat{z}_{1}\hat{z}_0,0,\cdots,0,0,\cdots,0   \nonumber\\
\hat{z}_{0}\hat{z}_{2},\hat{z}_{1}^{2},\hat{z}_{2}\hat{z}_{0},\cdots,0,0,\cdots,0 \nonumber \\
\vdots\nonumber\\
\hat{z}_{0}\hat{z}_{(N-1)/2},\hat{z}_{1}\hat{z}_{(N-1)/2-1},\hat{z}_{2}\hat{z}_{(N-1)/2-2},\cdots,\hat{z}_{(N-1)/2}\hat{z}_{0},0,\cdots,0  \label{odd-pyramid}\\
0,\hat{z}_{1}\hat{z}_{(N-1)/2},\hat{z}_{2}\hat{z}_{(N-1)/2-1},\cdots,\hat{z}_{(N-1)/2}\hat{z}_{1},0,\cdots,0   \nonumber\\
\vdots \nonumber\\
0,0,0,\cdots,\hat{z}_{(N-1)/2-1}\hat{z}_{(N-1)/2},\hat{z}_{(N-1)/2}\hat{z}_{(N-1)/2-1},0,\cdots,0 \nonumber\\
0,0,0,\cdots,0,\hat{z}^{2}_{(N-1)/2},0,\cdots,0.\nonumber
\end{gather}
\end{center}
It takes the identical form as that in \cite[(3.4)]{Eldar} for $B$-bandlimited signals  where  $B\leq N/2$.
Therefore,  by letting $\arg(\hat{z}_{0})\in \{0, \pi\}$ and  using  the similar procedures in \cite[section 3.2]{Eldar}, we can determine   $\textbf{z}$ up to the ambiguity in
Proposition \ref{oddlest} by FROG measurements. Thus  this paper will be only focused the case when $N$  is even.
\vspace{2mm}

 {\bf Case II:  $N$ is even.}  Suppose that $\textbf{z} \in \mathbb{C}^{N}$ such that $N$ is even. Then the corresponding  system \eqref{FROGFREQ} has  the following  structure w.r.t   variables $\hat{z}_{0}, \ldots, \hat{z}_{N/2}$:
\begin{center}
\begin{gather}
\hat{z}_{0}^{2},0,0,\cdots,\hat{z}_{N/2}^{2},0,\cdots,0 \nonumber \\
\hat{z}_{0}\hat{z}_{1},\hat{z}_{1}\hat{z}_0,0,\cdots,0,0,\cdots,0 \nonumber \\
\hat{z}_{0}\hat{z}_{2},\hat{z}_{1}^{2},\hat{z}_{2}\hat{z}_{0},\cdots,0,0,\cdots,0 \nonumber \\
\vdots\nonumber\\
\hat{z}_{0}\hat{z}_{N/2-1},\hat{z}_{1}\hat{z}_{N/2-2},\hat{z}_{2}\hat{z}_{N/2-3},\cdots,0,0,\cdots,0   \label{even-pyramid}\\
\hat{z}_{0}\hat{z}_{N/2},\hat{z}_{1}\hat{z}_{N/2-1},\hat{z}_{2}\hat{z}_{N/2-2},\cdots,\hat{z}_{N/2}\hat{z}_{0},0,\cdots,0   \nonumber\\
0,\hat{z}_{1}\hat{z}_{N/2},\hat{z}_{2}\hat{z}_{N/2-1},\cdots,\hat{z}_{N/2}\hat{z}_{1},0,\cdots,0   \nonumber\\
\vdots \nonumber\\
0,0,\cdots,0,\hat{z}_{N/2-1}\hat{z}_{N/2},\hat{z}_{N/2}\hat{z}_{N/2-1},0,\cdots,0.  \nonumber
\end{gather}
\end{center}
Note that if $\hat{z}_{N/2}\neq0$,  then \eqref{even-pyramid} (actually not a pyramid structure) takes the different form from \eqref{odd-pyramid}.
The  procedures for  FROG-PR in \cite[section 3.2]{Eldar} do not hold for $\textbf{z}$.
Firstly, in \cite[section 3.2]{Eldar}, for a $B$-bandlimited signal    $\textbf{y}\in \mathbb{C}^{N}$ such that $B\leq N/2$, $\arg(\hat{y}_{0})$ and $\arg(\hat{y}_{1})$
can be assigned arbitrarily. Such an assignment does not  hold for the  analytic signal $\textbf{z}$  since
Proposition \ref{even} implies that the  arbitrary translation and rotation do not necessarily  lead to the ambiguity of  FROG-PR.
Secondly, unlike those  in \eqref{odd-pyramid} and \cite[(3.4)]{Eldar}, the first row in \eqref{even-pyramid}
is involved with the  two variables $\hat{z}_{0}, \hat{z}_{N/2}\in \mathbb{R}$. Then $\hat{z}_{0}$ can not be determined by just letting $\arg(\hat{z}_{0})\in \{0, \pi\}$. Actually, it will be clear  in Theorems \ref{auxiliarytheorem} and \ref{yinli2.1} that the  determination of
$\hat{z}_{0}$ is absolutely not trivial when $|\hat{z}_{0}|\neq |\hat{z}_{N/2}|$. On the other hand,
from the perspective of $L$ in \eqref{frogmeasuement}, the problem of  FROG-PR of analytic signals of even lengths is also
different from that of bandlimited signals in \cite{Eldar} since it is required in this paper that $L$ needs to be odd.
More details will be included in section \ref{future}.
\section{The Main Results}\label{section3}

The aim of this section is to establish the uniqueness results for  the  FROG-PR problem for   analytic signals in
$\mathbb{C}^{N}$ by the $[N, L]$-FROG measurements, where   $N$ and $L$ are   respectively even and  odd such that
 $r=\lceil N/L\rceil\geq 5$. This will be achieved by introducing
a series of  approaches/algorithms for the determination of such analytic signals.


From now on
the DFT of an analytic signal  $\textbf{z}\in \mathbb{C}^{N}$ is denoted by $\widehat{\textbf{z}}=(\hat{z}_{0}, \ldots, \hat{z}_{N/2},$
$ 0, \ldots, 0)$.
It follows from  \eqref{FROGFREQ} and  $\hbox{supp}(\widehat{\textbf{z}})\subseteq\{0, \ldots, N/2\}$
that, the  FROG-PR of $\textbf{z}$  is equivalent to finding  an analytic signal $\tilde{\textbf{z}}\in \mathbb{C}^{N}$
such that its DFT $\widehat{\tilde{\textbf{z}}}=(\hat{\tilde{z}}_{0}, \ldots,  \hat{\tilde{z}}_{N/2}, \overbrace{0, \ldots, 0}^{N/2-1})$ satisfies the following conditions:
\begin{align}\label{FROGFREQ123489}\left\{\begin{array}{lllll}
 |\hat{y}_{0,m}|&=&\frac{1}{N}\big|\hat{\tilde{z}}^{2}_{0}+\hat{\tilde{z}}^{2}_{N/2}w^{Nm/2}\big|, \quad  (\ref{FROGFREQ123489} \hbox{A})\\
 |\hat{y}_{k,m}|&=&\frac{1}{N}\big|\sum_{l=0}^{k}\hat{\tilde{z}}_{l}\hat{\tilde{z}}_{k-l}w^{lm}\big|, k=1,\cdots, N/2, \quad  (\ref{FROGFREQ123489} \hbox{B})\\
 |\hat{y}_{k,m}|&=&\frac{1}{N}\big|\sum_{l=k-N/2}^{N/2}\hat{\tilde{z}}_{l}\hat{\tilde{z}}_{k-l}w^{lm}\big|, k= N/2+1,\cdots, N-1, \quad  (\ref{FROGFREQ123489} \hbox{C})\\
 &&m=0,1,\cdots,r-1, \end{array}\right. 
  \end{align}
where  $\{|\hat{y}_{k,m}|^{2}\}$ are  the $[N, L]$-FROG measurements of $\textbf{z}$.
Clearly, $\widehat{\textbf{z}}$ satisfies \eqref{FROGFREQ123489}.
\vspace{3mm}


We outline  below  the four key steps that  we will use to find such a  $\tilde{\textbf{z}}$.
\vspace{3mm}

\begin{itemize}


\item[(i)] It follows from Proposition \ref{characterizationeven} (i) that $\hat{z}_{0}$ is real-valued.
Subsection \ref{determinationofz0} will be used  to determine   $\hat{\tilde{z}}_{0}\in \mathbb{R}$  up to a sign.


\item[(ii)] It follows from (\ref{FROGFREQ123489}\hbox{B}) that     $|\hat{\tilde{z}}_{1}|:=\frac{N|\hat{y}_{1,0}|}{2|\hat{\tilde{z}}_{0}|}$.
Setting $\hat{\tilde{z}}_{0}=\epsilon\hat{z}_{0}$  with $\epsilon\in \{1, -1\}$  and $\hat{\tilde{z}}_{1}=\frac{N|\hat{y}_{1,0}|}{2|\hat{\tilde{z}}_{0}|}$,
Subsection \ref{determinationofz1} will be devoted to finding   a  signal $\tilde{\textbf{z}}\in \mathbb{C}^{N}$
such that its DFT $\widehat{\tilde{\textbf{z}}}=(\hat{\tilde{z}}_{0}, \hat{\tilde{z}}_{1}, \ldots,  \hat{\tilde{z}}_{N/2}, \overbrace{0, \ldots, 0}^{N/2-1})=
(\epsilon\hat{z}_{0}, \frac{N|\hat{y}_{1,0}|}{2|\hat{\tilde{z}}_{0}|},  \ldots,  \hat{\tilde{z}}_{N/2}, \overbrace{0, \ldots, 0}^{N/2-1})$ satisfies
\begin{align}\label{clearly}\left\{\begin{array}{lllll}
 |\hat{y}_{k,m}|&=&\frac{1}{N}\big|\sum_{l=0}^{k}\hat{\tilde{z}}_{l}\hat{\tilde{z}}_{k-l}w^{lm}\big|, k=\textcolor[rgb]{0.00,0.07,1.00}{2},\cdots, N/2,\\
 |\hat{y}_{k,m}|&=&\frac{1}{N}\big|\sum_{l=k-N/2}^{N/2}\hat{\tilde{z}}_{l}\hat{\tilde{z}}_{k-l}w^{lm}\big|, k= N/2+1,\cdots, N-1,\\
 &&m=0,1,\cdots,r-1. \end{array}\right.  
  \end{align}
Clearly,  \eqref{clearly} is the relaxed form of  \eqref{FROGFREQ123489}.


\item[(iii)] Based on the result of  (ii) , subsection \ref{subsectionadd} concerns on the solution to the  following equation system  w.r.t    $\tilde{\textbf{z}}\in \mathbb{C}^{N}$
such that its DFT $\widehat{\tilde{\textbf{z}}}=(\hat{\tilde{z}}_{0}, \ldots,  \hat{\tilde{z}}_{N/2}, \overbrace{0, \ldots, 0}^{N/2-1})$:
\begin{align}\label{fromone}\left\{\begin{array}{lllll}
 |\hat{y}_{k,m}|&=&\frac{1}{N}\big|\sum_{l=0}^{k}\hat{\tilde{z}}_{l}\hat{\tilde{z}}_{k-l}w^{lm}\big|, k=\textcolor[rgb]{0.00,0.07,1.00}{1},\cdots, N/2,\\
 |\hat{y}_{k,m}|&=&\frac{1}{N}\big|\sum_{l=k-N/2}^{N/2}\hat{\tilde{z}}_{l}\hat{\tilde{z}}_{k-l}w^{lm}\big|, k= N/2+1,\cdots, N-1,\\
 &&m=0,1,\cdots,r-1. \end{array}\right.  
  \end{align}
Clearly,  \eqref{fromone}  is also the relaxed form of  \eqref{FROGFREQ123489} but  \textbf{less}  relaxed than  \eqref{clearly} since it also contains the $k=1$ case.


\item[(iv)] Having  $\tilde{\textbf{z}}\in \mathbb{C}^{N}$ satisfying  \eqref{fromone} at hand, the procedure in  (translation technique-based)  Approach \ref{APPROACH3} in  subsection  \ref{section2.5}  will allow us to determine  the solution to \eqref{FROGFREQ123489}.
\end{itemize}

\subsection{Auxiliary results}

\begin{lem}\label{lemmm2.1} (\cite{Eldar})
Consider the equation system  w.r.t $z\in \mathbb{C}$:
\begin{align} \label{fangcheng1}   \left\{
\begin{aligned}
\big|z+v_{1}\big|=n_{1}, \\
\big|z+v_{2}\big|=n_{2}, \\
\big|z+v_{3}\big|=n_{3},
\end{aligned}
\right.
\end{align}
where $n_{1},n_{2},n_{3}\geq 0$,  and  $v_{1},v_{2},v_{3}\in \mathbb{C}$ are pairwise  distinct. Suppose that  there exists a solution $\mathring{z} = a+\textbf{i}b$ to \eqref{fangcheng1}. If
 \begin{align}\label{condition}  \Im\Big\{\frac{v_{1}-v_{2}}{v_{1}-v_{3}}\Big\}\ne0,\end{align} then $\mathring{z}$ is the unique  solution and it is  given through
\begin{gather}\label{qiujie1}
\begin{pmatrix}a \\b\end{pmatrix}=\frac{1}{2}\left(\begin{array}{cccccccccc}c & -d \\ e & -f\end{array}\right)^{-1}\begin{pmatrix}n_{1}^{2}-n_{2}^{2}+|v_{2}|^{2}-|v_{1}^{2}| \\ n_{1}^{2}-n_{3}^{2}+|v_{3}|^{2}-|v_{1}^{2}| \end{pmatrix},
\end{gather} where $v_{1}-v_{2}=c-\textbf{i}d$ and $v_{1}-v_{3}=e-\textbf{i}f$.
\end{lem}

Comparing with Lemma \ref{lemmm2.1}, the following system of two equations has more than one solution.

\begin{lem}\label{lemmm2.2}
Consider the equation system  w.r.t $z\in \mathbb{C}$:
\begin{align} \label{fangcheng2}   \left\{
\begin{aligned}
\big|z+mv_{1}\big|=n_{1}, \\
\big|z+mv_{2}\big|=n_{2},
\end{aligned}
\right.
\end{align}
where $n_{1},n_{2}\geq 0$,   $v_{1},v_{2}\in \mathbb{C}$ are   distinct and $0\neq m\in \mathbb{C}$. Suppose that  there exists a solution $\mathring{z} = a+\textbf{i}b$ to \eqref{fangcheng2}.

(i) If     $m, v_{1},v_{2}\in \mathbb{R}$, then  the other solution  is  $\bar{\mathring{z}}$ and it is  given through
 \begin{align}\label{qiujie2}\left\{\begin{array}{lllll}
 a&=&(n_{1}^{2}-n_{2}^{2})/[2m(v_{1}-v_{2})],\\
 |b|&=&\sqrt{n_{1}^{2}-(a+mv_{1})^{2}}.
  \end{array}\right.\end{align}

(ii) If     $m\in \mathbb{C}\setminus\mathbb{R}$ and   $v_{1},v_{2}\in \mathbb{R}$, then there exits another solution
$z'\neq\bar{\mathring{z}}$. Moreover,   $\mathring{z}$ and $z'$ can be given by
\begin{align}\label{qiujie3}
\mathring{z}=m(\mathring{a}+\textbf{i}\mathring{b}),z'=m(\mathring{a}-\textbf{i}\mathring{b})
\end{align}
where $\mathring{a},\mathring{b} \in \mathbb{R}$ are given by
\begin{align}\left\{\begin{array}{lllll}
\mathring{a}&=&(|\frac{n_{1}}{m}|^{2}-|\frac{n_{2}}{m}|^{2})/[2(v_{1}-v_{2})],\\
|\mathring{b}|&=&\sqrt{|\frac{n_{1}}{m}|^{2}-(a+v_{1})^{2}}.
  \end{array}\right.\end{align}
\end{lem}
\begin{proof}
Item  (i)  is  derived from the last paragraph of the proof of \cite[Lemma 3.2]{Eldar}. We just need to prove item (ii).
Denote   $\frac{z}{m}=\mathring{a}+\textbf{i}\mathring{b}$. Then it follows from $v_{1},v_{2}\in \mathbb{R}$ that
\eqref{fangcheng2} is equivalent to
\begin{align}\label{fhhh1} \left\{
\begin{aligned}
(\mathring{a}+v_{1})^{2}+\mathring{b}^{2}=|\frac{n_{1}}{m}|^{2}, \\
(\mathring{a}+v_{2})^{2}+\mathring{b}^{2}=|\frac{n_{2}}{m}|^{2}.
\end{aligned}
\right.
\end{align}
 Through the  direct calculation we have
$\mathring{a}=(\big|\frac{n_{1}}{m}\big|^{2}-\big|\frac{n_{2}}{m}\big|^{2})/[2(v_{1}-v_{2})].$  On the other hand, if $(\mathring{a}, \mathring{b})$ is the solution to \eqref{fhhh1}, then the other solution
is $(\mathring{a}, -\mathring{b})$.
%
%
%
Stated another way, the solutions to \eqref{fangcheng2} are  $\mathring{z}=m(\mathring{a}+\textbf{i}\mathring{b})$ and $z'=m(\mathring{a}-\textbf{i}\mathring{b})$. Since $m\in \mathbb{C}\setminus\mathbb{R}$, then $z'\neq
\bar{\mathring{z}}.$
\end{proof}

To be clear in Remarks \ref{remark31} and \ref{remark22} that, the following lemma will be needed for the existence of the FROG measurements in Approaches  \ref{bvc} and \ref{approach22}.

\begin{lem}\label{lemma233}
If $5\leq r\in \mathbb{N}$, then  there exist $m_{1}, \ldots, m_{5}\subseteq\{0, 1, \ldots, r-1\}$
such that (1) $1+w^{2m_{l}}\ne0$  ($l=1, \ldots, 5$) and (2) $\frac{w^{m_{1}}}{1+w^{2m_{1}}}\ne\frac{1}{2}$, where $w=e^{\frac{\textbf{i}2\pi}{r}}$. Moreover, (3) there exists $i \in \{1,2,\ldots,r-1\}$ such that   $\frac{w^{i}+w^{2i}}{1+w^{3i}}\ne1$,
\end{lem}

\begin{proof}
The proof will be given in section \ref{kkeke}.
\end{proof}

The following theorem will be used  in section \ref{subsectionadd}.






\begin{theo}\label{auxiliarytheorem}
Suppose that $\textbf{z}\in \mathbb{C}^{N}$ is a  generic  analytic signal such that $N$ is even, and its DFT
\begin{equation}
\widehat{\textbf{z}}=(\hat{z}_{0}, \hat{z}_{1}, \ldots, \hat{z}_{N/2}, 0,
\ldots, 0):=(|\hat{z}_{0}|e^{\textbf{i}\theta_{0}}, |\hat{z}_{1}|e^{\textbf{i}\theta_{1}}, \ldots, |\hat{z}_{N/2}|e^{\textbf{i}\theta_{N/2}}, 0, \ldots, 0)
\end{equation}
satisfies $|\hat{z}_{0}|\neq0 $ and $|\hat{z}_{1}|\neq0.$
Consider  the following  equation system  w.r.t $\tilde{\textbf{z}}\in \mathbb{C}^{N}$ such that $\widehat{\tilde{\textbf{z}}}=(\hat{\tilde{z}}_{0}, \ldots,  \hat{\tilde{z}}_{N/2}, \overbrace{0, \ldots, 0}^{N/2-1})$:
\begin{align}\label{formore1234}\left\{\begin{array}{lllll}
 |\hat{y}_{k,m}|&=&\frac{1}{N}\big|\sum_{l=0}^{k}\hat{\tilde{z}}_{l}\hat{\tilde{z}}_{k-l}w^{lm}\big|, k=1,\cdots, N/2,\\
 |\hat{y}_{k,m}|&=&\frac{1}{N}\big|\sum_{l=k-N/2}^{N/2}\hat{\tilde{z}}_{l}\hat{\tilde{z}}_{k-l}w^{lm}\big|, k= N/2+1,\cdots, N-1,\\
 &&m\in \Lambda, \end{array}\right.  
  \end{align}
where   $|\hat{y}_{k,m}|^{2}$ is the $[N, L]$-FROG measurement
of $\textbf{z}$,  $r=\lceil N/L \rceil\geq5$ and $\Lambda\subseteq\{0, 1, \ldots, r-1\}$ is arbitrary  such that  $\#\Lambda=5$. If  $|\hat{\tilde{z}}_{0}|\neq |\hat{z}_{0}|$ then  $\tilde{\textbf{z}}$ does not satisfy  \eqref{formore1234}.
\end{theo}
\begin{proof}
Suppose that $\hat{\tilde{z}}_{0}=\lambda|\hat{z}_{0}|e^{\textbf{i}\alpha}$ such that $1\neq\lambda \geq 0$ and $\alpha\in \mathbb{R}$ is arbitrary. If $\hat{\tilde{z}}_{0}=0$ or $\hat{\tilde{z}}_{1}=0$, then it is easy to check that
FROG measurements at $(1, 0)$ of $\tilde{\textbf{z}}$ and those of  $\textbf{z}$ are not  identical, and the proof is concluded.
Otherwise, we have $\lambda>0$ and
%
%
%
%
\begin{equation}\label{98nbv}
\frac{2}{N}\big|\hat{z}_{0}\hat{z}_{1}\big|=|\hat{y}_{1,0}|=\frac{2}{N}\big|\hat{\tilde{z}}_{0}\hat{\tilde{z}}_{1}\big|,
\end{equation}
which implies 
$
|\hat{\tilde{z}}_{1}|=\big|\frac{\hat{z}_{1}}{\lambda}\big|.$ Now let $\hat{\tilde{z}}_{1}=|\frac{\hat{z}_{1}}{\lambda}|e^{\textbf{i}\theta}$ such that $\theta\in\mathbb{R} $ is arbitrary. If the FROG measurements at $(2, s)$ of $\tilde{\textbf{z}}$ and $\textbf{z}$
are not identical,  then the proof is concluded, where  $s\in \Lambda$. Otherwise,  we have
\begin{equation}\label{hxm}
\frac{1}{N}\big|\hat{z}_{0}\hat{z}_{2}+\hat{z}_{1}^{2}w^{m}+\hat{z}_{2}\hat{z}_{0}w^{2m}\big|=
|\hat{y}_{2,m}|=\frac{1}{N}\big|\hat{\tilde{z}}_{0}\hat{\tilde{z}}_{2}
+\hat{\tilde{z}}_{1}^{2}w^{m}+\hat{\tilde{z}}_{2}\hat{\tilde{z}}_{0}w^{2m}\big|, m\in \Lambda.
\end{equation}
Since $\hat{\tilde{z}}_{0}=\lambda|\hat{z}_{0}|e^{\textbf{i}\alpha}$ and  $\hat{\tilde{z}}_{1}=|\frac{\hat{z}_{1}}{\lambda}|e^{\textbf{i}\theta}$,
then \eqref{hxm} can be expressed as
\begin{align}\label{986}
\big|\hat{z}_{0}\hat{z}_{2}(1+w^{2m})+\hat{z}_{1}^{2}w^{m}\big|^{2}=\Big|\lambda|\hat{z}_{0}|e^{\textbf{i}\alpha}\hat{\tilde{z}}_{2}(1+w^{2m})+|\frac{\hat{z}_{1}}{\lambda}|^{2}e^{\textbf{i}2\theta}w^{m}\Big|^{2}.
\end{align}
It is easy to check that  \eqref{986} is equivalent to
\begin{align}\label{987}\begin{array}{llll}
0=&(\lambda^{2}|\hat{\tilde{z}}_{2}|^{2}-|\hat{z}_{2}|^{2})|\hat{z}_{0}|^{2}|1+w^{2m}|^{2}+(\frac{1}{\lambda^{4}}-1)|\hat{z}_{1}|^{4} \\
  &+2\Re\Big\{\overline{(1+w^{2m})}w^{m}(\overline{\lambda|\hat{z}_{0}|e^{\textbf{i}\alpha}\hat{\tilde{z}}_{2}}|\frac{\hat{z}_{1}}{\lambda}|^{2}e^{\textbf{i}2\theta}
  -\overline{\hat{z}_{0}\hat{z}_{2}}\hat{z}_{1}^{2})\Big\},
\end{array}\end{align}
where  $m\in \Lambda$.
Multiplying both sides of \eqref{987} by $w^{2m}$   leads to
\begin{align}\label{988}\begin{array}{llll}
0=(\lambda^{2}|\hat{\tilde{z}}_{2}|^{2}-|\hat{z}_{2}|^{2})\hat{z}_{0}^{2}(1+w^{2m})^{2}+(\frac{1}{\lambda^{4}}-1)|\hat{z}_{1}|^{4}w^{2m} \\
+\big[(\overline{\lambda|\hat{z}_{0}|e^{\textbf{i}\alpha}\hat{\tilde{z}}_{2}}|\frac{\hat{z}_{1}}{\lambda}|^{2}e^{\textbf{i}2\theta}
-\lambda|\hat{z}_{0}|e^{\textbf{i}\alpha}\hat{\tilde{z}}_{2}\overline{|\frac{\hat{z}_{1}}{\lambda}|^{2}e^{\textbf{i}2\theta}})
  -(\overline{\hat{z}_{0}\hat{z}_{2}}\hat{z}_{1}^{2}-\hat{z}_{0}\hat{z}_{2}\overline{\hat{z}_{1}^{2}})\big](1+w^{2m})w^{m}.
\end{array}\end{align}
Consider the following equation w.r.t $x$:
\begin{align}\label{989}\begin{array}{llll}
0=&(\lambda^{2}|\hat{\tilde{z}}_{2}|^{2}-|\hat{z}_{2}|^{2})\hat{z}_{0}^{2}(1+x^{2})^{2}+(\frac{1}{\lambda^{4}}-1)|\hat{z}_{1}|^{4}x^{2} \\
&+\big[(\overline{\lambda|\hat{z}_{0}|e^{\textbf{i}\alpha}\hat{\tilde{z}}_{2}}|\frac{\hat{z}_{1}}{\lambda}|^{2}e^{\textbf{i}2\theta}
-\lambda|\hat{z}_{0}|e^{\textbf{i}\alpha}\hat{\tilde{z}}_{2}\overline{|\frac{\hat{z}_{1}}{\lambda}|^{2}e^{\textbf{i}2\theta}})
  -(\overline{\hat{z}_{0}\hat{z}_{2}}\hat{z}_{1}^{2}-\hat{z}_{0}\hat{z}_{2}\overline{\hat{z}_{1}^{2}})\big](1+x^{2})x.
\end{array}\end{align}
Clearly, if the polynomial on the right-hand side of  \eqref{989} is not a zero polynomial, then it has  at most $4$ solutions.
By  \eqref{988},   $w^{m}(m\in \Lambda)$ are the $\#\Lambda$ solutions to  \eqref{989}. Since $\#\Lambda\geq5 $, then all the coefficients in \eqref{989} are zero. That is,
$$ \left\{
\begin{aligned}
&(\lambda^{2}|\hat{\tilde{z}}_{2}|^{2}-|\hat{z}_{2}|^{2})\hat{z}_{0}^{2}=0, \\
&(\frac{1}{\lambda^{4}}-1)|\hat{z}_{1}|^{4}=0, \\
&\big(\overline{\lambda|\hat{z}_{0}|e^{\textbf{i}\alpha}\hat{\tilde{z}}_{2}}|\frac{\hat{z}_{1}}{\lambda}|^{2}e^{\textbf{i}2\theta}
-\lambda|\hat{z}_{0}|e^{\textbf{i}\alpha}\hat{\tilde{z}}_{2}\overline{|\frac{\hat{z}_{1}}{\lambda}|^{2}e^{\textbf{i}2\theta}}\big)
  -\big(\overline{\hat{z}_{0}\hat{z}_{2}}\hat{z}_{1}^{2}-\hat{z}_{0}\hat{z}_{2}\overline{\hat{z}_{1}^{2}}\big)=0.
\end{aligned}
\right.
$$
Since $|\hat{z}_{1}|\neq0$ then $\frac{1}{\lambda^{4}}-1=0$ and $\lambda=1$, which    contradicts with
 the previous assumption  $1\neq\lambda \geq 0$.
This completes the proof.
\end{proof}

\subsection{Determination of $\hat{\tilde{z}}_{0}$}\label{determinationofz0}
In what follows, we establish an approach to determining $\hat{\tilde{z}}_{0}$ in \eqref{FROGFREQ123489}   up to a sign.
Its theoretic guarantee will be presented in Theorem \ref{yinli2.1}.
\begin{appr}\label{bvc}
\textbf{Input}: $[N, L]$-FROG measurements $\Big\{|\hat{y}_{0,0}|^{2}, |\hat{y}_{0,1}|^{2}, |\hat{y}_{1,0}|^{2}, |\hat{y}_{2,i_{2}^{(q)}}|^{2} \big|:i_{2}^{(1)}=0, w^{2i_{2}^{(q)}}\ne-1, \frac{w^{i_{2}^{(2)}}}{1+w^{2i_{2}^{(2)}}}\ne\frac{1}{2},0\le i_{2}^{(q)}\le r-1,q=1,2, \ldots,5\Big\}$.
\%\% $w=e^{\textbf{i}2\pi/r}$ and  $r=\lceil N/L\rceil.$
\newline
\textbf{Step 1}: If $|\hat{y}_{0,1}|=0$, then $\hat{\tilde{z}}_{0}=\sqrt{\frac{N|\hat{y}_{0,0}|}{2}}$ and we terminate the program.
If not, then  conduct  Step 2 and Step 3 to find $\hat{\tilde{z}}_{0}$.
\newline
\textbf{Step 2}: $\hat{\tilde{z}}_{0}\gets \sqrt{\frac{N(|\hat{y}_{0,0}|+|\hat{y}_{0,1}|)}{2}}$; $\hat{\tilde{z}}_{1}\gets \frac{N|\hat{y}_{1,0}|}{2\hat{\tilde{z}}_{0}}$;
\newline
\textbf{Step 3}:
If the following equation system  w.r.t $\hat{\tilde{z}}_{2}$:
\begin{align}\label{check123}
\frac{1}{N}|\hat{\tilde{z}}_{0}\hat{\tilde{z}}_{2}(1+w^{2i_{2}^{(q)}})+\hat{\tilde{z}}_{1}^{2}w^{i_{2}^{(q)}}|=
|\hat{y}_{2,i_{2}^{(q)}}|, q=1,2,\ldots,5
\end{align}
does not have  a solution, then  $\hat{\tilde{z}}_{0}\gets\sqrt{\frac{N(|\hat{y}_{0,0}|-|\hat{y}_{0,1}|)}{2}}$.\\
\textbf{Output}: $\hat{\tilde{z}}_{0}$.
\end{appr}

\begin{rem}\label{remark31}(1) By Lemma \ref{lemma233}, the requirement: $w^{2i_{2}^{(q)}}\ne-1, \frac{w^{i_{2}^{(2)}}}{1+w^{2i_{2}^{(2)}}}\ne\frac{1}{2}$
in Approach \ref{bvc} can be satisfied.

(2) Clearly, \eqref{check123} is equivalent to  the condition that $\hat{\tilde{z}}_{2}$ lies in the $5$  circles on $\mathbb{C}$: $$\Big\{z: \Big|z-(-\hat{\tilde{z}}_{1}^{2}w^{i_{2}^{(q)}}/(\hat{\tilde{z}}_{0}(1 +w^{2i_{2}^{(q)}})))\Big|=N|\hat{y}_{2,i_{2}^{(q)}}|\Big\}, \ q=1,2 \ldots,5.$$
By the correlations among the five  circles, it is easy to  check that whether
\eqref{check123} has a solution.
\end{rem}

The following theorem states that   $\hat{\tilde{z}}_{0}$ can be determined (up to a sign) by Approach \ref{bvc}.

\begin{theo}\label{yinli2.1}
Suppose  that  $\textbf{z}\in \mathbb{C}^{N}$ (with $N$ being  even)  is a generic   analytic  signal  such that  its DFT $\widehat{\textbf{z}}=(\hat{z}_{0}, \hat{z}_{1},$
$ \ldots, \hat{z}_{N/2},
\ldots, 0)$ satisfies  $|\hat{z}_{0}|\neq0 $ and $|\hat{z}_{1}|\neq0.$
  If $L$ is odd and  $r=\lceil N/L\rceil\geq5$,
  then $\hat{\tilde{z}}_{0}$ in \eqref{FROGFREQ123489} can be determined up to a sign by Approach \ref{bvc}.
\end{theo}

\begin{proof}
The proof will be presented in section \ref{proofof34}.
\end{proof}

\subsection{Finding   $\tilde{\textbf{z}}\in \mathbb{C}^{N}$
such that its DFT $\widehat{\tilde{\textbf{z}}}=(\epsilon\hat{z}_{0}, \frac{N|\hat{y}_{1,0}|}{2|\hat{\tilde{z}}_{0}|}, \hat{\tilde{z}}_{2} \ldots,  \hat{\tilde{z}}_{N/2}, 0, \ldots, 0)$
satisfies \eqref{clearly}}\label{determinationofz1}
Recall that $\hat{\tilde{z}}_{0}$ in \eqref{FROGFREQ123489} can be determined (up to a sign) by Approach \ref{bvc}, and  it follows from \ref{FROGFREQ123489}(B) that  $|\hat{\tilde{z}}_{1}|:=\frac{N|\hat{y}_{1,0}|}{2|\hat{\tilde{z}}_{0}|}$.
This subsection concerns on the solution to \eqref{clearly} w.r.t  $\tilde{\textbf{z}}$. For convenience, \eqref{clearly} is stated again as follows,
\begin{align}\label{localequation}\left\{\begin{array}{lllll}
 |\hat{y}_{k,m}|&=&\frac{1}{N}\big|\sum_{l=0}^{k}\hat{\tilde{z}}_{l}\hat{\tilde{z}}_{k-l}w^{lm}\big|, k=2,\cdots, N/2,\\
 |\hat{y}_{k,m}|&=&\frac{1}{N}\big|\sum_{l=k-N/2}^{N/2}\hat{\tilde{z}}_{l}\hat{\tilde{z}}_{k-l}w^{lm}\big|, k= N/2+1,\cdots, N-1,\\
 &&m=0,1,\cdots,r-1, \end{array}\right.  
  \end{align}
where $\hat{\tilde{z}}_{0}=\epsilon\hat{z}_{0}$ with $\epsilon\in \{1, -1\}$ and  $\hat{\tilde{z}}_{1}:=\frac{N|\hat{y}_{1,0}|}{2|\hat{\tilde{z}}_{0}|}$. In what follows, we establish an approach to  find such a
$\tilde{\textbf{z}}$  by
the $3N/2+1$ measurements: \begin{align}\label{FORGMEASUREMNTS}\begin{array}{lll}\Big\{|\hat{y}_{0,0}|^{2},|\hat{y}_{0,1}|^{2},|\hat{y}_{1,0}|^{2},|\hat{y}_{2,i_{2}^{(q)}}|^{2} ,|\hat{y}_{3,0}|^{2},|\hat{y}_{3,i_{3}}|^{2}
  ,|\hat{y}_{k,i_{2}^{(p)}}|^{2}\big|: i_{2}^{(1)}=0, w^{2i_{k}^{(q)}}\ne-1,
\frac{w^{i_{2}^{(2)}}}{1+w^{2i_{2}^{(2)}}}\ne\frac{1}{2},\\\frac{w^{i_{3}}+w^{2i_{3}}}{1+w^{3i_{3}}}
  \ne1,w^{ki_{k}^{(p)}}\ne-1, i_{k}^{(1)}=0,i_{k}^{(2)}+i_{k}^{(3)}\ne r, 0\le i_{2}^{(q)}\le r-1,1\le i_{3}\le r-1,\\0\le i_{k}^{(p)}\le r-1,
4\le k\le N/2, p=1,2,3,q=1,2,3,4,5\Big\}.
  \end{array}\end{align}
The  theoretic guarantee for such an approach  will be given in  Theorem \ref{localtheorem}. The existence of the above  measurements is addressed in the following remark.

 \begin{rem} \label{remark22} (1)
Recall that the FROG measurements in \eqref{FORGMEASUREMNTS} need to satisfy the requirements:
$$w^{2i_{k}^{(q)}}\ne-1,
\frac{w^{i_{2}^{(2)}}}{1+w^{2i_{2}^{(2)}}}\ne\frac{1}{2},  \frac{w^{i_{3}}+w^{2i_{3}}}{1+w^{3i_{3}}}
  \ne1, $$
which is guaranteed by  Lemma \ref{lemma233}.

(2)  According to  the
analysis  in \cite[Page 1038]{Eldar}, for any $k\in \{4, \ldots, N/2\}$ there exist $\{i_{k}^{(p)}:p=1,2,3\}\subset\{0,1,\ldots,r-1\}$ such that
$w^{ki_{k}^{(p)}}\ne-1, i_{k}^{(1)}=0 $ and $ i_{k}^{(2)}+i_{k}^{(3)}\ne r$.

 \end{rem}


\begin{appr}\label{approach22}
\textbf{Input:}
FROG measurements in \eqref{FORGMEASUREMNTS}, $\hat{\tilde{z}}_{0}$ (derived from Approach \ref{bvc})
and $\hat{\tilde{z}}_{1}= \frac{N|\hat{y}_{1,0}|}{2|\hat{\tilde{z}}_{0}|}$.\\
\textbf{Step 1}: By Lemma \ref{lemmm2.2} (i), choose  a solution $\hat{\tilde{z}}_{2}$ to
$$ \left\{
  \begin{aligned}
  |\hat{y}_{2,0}|&=\frac{1}{N}\big|2\hat{\tilde{z}}_{0}\hat{\tilde{z}}_{2}+\hat{\tilde{z}}_{1}^{2}\big|, \\
  |\hat{y}_{2,i_{2}^{(2)}}|&=\frac{1}{N}\big|\hat{\tilde{z}}_{0}\hat{\tilde{z}}_{2}+\hat{\tilde{z}}_{1}^{2}w^{i_{2}^{(2)}}
  +\hat{\tilde{z}}_{2}\hat{\tilde{z}}_{0}w^{2i_{2}^{(2)}}\big|.
  \end{aligned}
  \right.
  $$
\textbf{Step 2}: Given  $(\hat{\tilde{z}}_{0}, \hat{\tilde{z}}_{1}, \hat{\tilde{z}}_{2})$. Use  Lemma  \ref{lemmm2.2} (ii) and Lemma  \ref{lemmm2.1} to find the solution to
the system w.r.t $\hat{\tilde{z}}_{3},\hat{\tilde{z}}_{4}$:
\begin{align}\left\{\begin{array}{llllll}
  |\hat{y}_{3,0}|&=\frac{1}{N}\big|2\hat{\tilde{z}}_{0}\hat{\tilde{z}}_{3}+2\hat{\tilde{z}}_{1}\hat{\tilde{z}}_{2}\big|, \\
  |\hat{y}_{3,i_{3}}|&=\frac{1}{N}\big|\hat{\tilde{z}}_{0}\hat{\tilde{z}}_{3}+\hat{\tilde{z}}_{1}\hat{\tilde{z}}_{2}w^{i_{3}}
  +\hat{\tilde{z}}_{1}\hat{\tilde{z}}_{2}w^{2i_{3}}+\hat{\tilde{z}}_{3}\hat{\tilde{z}}_{0}w^{3i_{3}}\big|,\\
   |\hat{y}_{4,i_{4}^{(p)}}|&=\frac{1}{N}\big|\hat{\tilde{z}}_{0}\hat{\tilde{z}}_{4}+\hat{\tilde{z}}_{1}\hat{\tilde{z}}_{3}w^{i_{4}^{(p)}}+\hat{\tilde{z}}_{2}^{2}
  w^{2i_{4}^{(p)}}+\hat{\tilde{z}}_{3}\hat{\tilde{z}}_{1}w^{3i_{4}^{(p)}}+\hat{\tilde{z}}_{4}\hat{\tilde{z}}_{0}w^{4i_{4}^{(p)}}\big|,p=1,2,3.
 \end{array}\right.
 \end{align}
\textbf{Step 3}:  Given  $(\hat{\tilde{z}}_{0}, \hat{\tilde{z}}_{1}, \hat{\tilde{z}}_{2}, \hat{\tilde{z}}_{3}, \hat{\tilde{z}}_{4})$.
Use  Lemma  \ref{lemmm2.1}  to find iteratively  the solution to
the system w.r.t $\hat{\tilde{z}}_{j}, j\geq5$:
  \begin{align}
  |\hat{y}_{j,i_{j}^{(p)}}|&=\frac{1}{N}\big|\hat{\tilde{z}}_{0}\hat{\tilde{z}}_{j}+\hat{\tilde{z}}_{1}\hat{\tilde{z}}_{j-1}w^{i_{j}^{(p)}}+\ldots+
  \hat{\tilde{z}}_{j}\hat{\tilde{z}}_{0}w^{ji_{j}^{(p)}}\big|,p=1,2,3.
  \end{align}
\textbf{Output:} $\widehat{\tilde{\textbf{z}}}:=(\hat{\tilde{z}}_{0},\hat{\tilde{z}}_{1},\cdots,\hat{\tilde{z}}_{N/2},0,\cdots,0)$ and its IDFT $\tilde{\textbf{z}}$.
\end{appr}

The following theorem guarantees that $\tilde{\textbf{z}}$ derived from  Approach \ref{approach22}
is the solution to \eqref{localequation}.

\begin{theo}\label{localtheorem}
Suppose that   $\textbf{z}\in \mathbb{C}^{N}$ is  a generic   analytic signal
such that its  DFT
\begin{align} \label{ccfurrt1} \widehat{\textbf{z}}=(\hat{z}_{0}, \hat{z}_{1}, \ldots, \hat{z}_{N/2}, 0,
\ldots, 0):=(|\hat{z}_{0}|e^{\textbf{i}\theta_{0}}, |\hat{z}_{1}|e^{\textbf{i}\theta_{1}}, \ldots, |\hat{z}_{N/2}|e^{\textbf{i}\theta_{N/2}}, 0,
\ldots, 0).\end{align}
Let  $\hat{\tilde{z}}_{0}=\epsilon\hat{z}_{0}$ and $\hat{\tilde{z}}_{1}=\frac{N|\hat{y}_{1,0}|}{2|\hat{\tilde{z}}_{0}|}$ with $\epsilon\in \{1, -1\}$.
 Then the solutions to  \eqref{localequation}
are the vector  $(\mathfrak{E}(2)|\hat{z}_{2}|e^{\textbf{i}(\theta_{2}-2\theta_{1})}, \mathfrak{E}(3)|\hat{z}_{3}|e^{\textbf{i}(\theta_{3}-3\theta_{1})},\ldots,$
$\mathfrak{E}(N/2)|\hat{z}_{N/2}|e^{\textbf{i}[\theta_{N/2}-(N/2)\theta_{1}]})$  and
its complex conjugate,
where $\mathfrak{E}(k)$ takes $\epsilon$ and $1$ for $k$ being even and odd, respectively.
Moreover, one of the two solutions can be determined through Approach \ref{approach22}.
\end{theo}
\begin{proof}
The proof will be presented in subsection \ref{proofoftheorem2.5}.
\end{proof}


\subsection{On the solution to \eqref{fromone}}\label{subsectionadd}
As noted at the beginning of section \ref{section3},
our final  approach for FROG-PR is derived from the two relaxed FROG-PR problems in \eqref{clearly} (\eqref{localequation}) and \eqref{fromone}.
Recall that \eqref{clearly} has been addressed in Theorem \ref{localtheorem}. We now address
\eqref{fromone} with the following theorem. 

\begin{theo}\label{lemma2.6}
Suppose that $\textbf{z}\in \mathbb{C}^{N}$ is a generic   analytic signal such that $N$ is even, and its DFT
\begin{equation}
\widehat{\textbf{z}}=(\hat{z}_{0}, \hat{z}_{1}, \ldots, \hat{z}_{N/2}, 0,
\ldots, 0):=(|\hat{z}_{0}|e^{\textbf{i}\theta_{0}}, |\hat{z}_{1}|e^{\textbf{i}\theta_{1}}, \ldots, |\hat{z}_{N/2}|e^{\textbf{i}\theta_{N/2}}, 0, \ldots, 0)
\end{equation}
satisfies $|\hat{z}_{0}|\neq0 $ and $|\hat{z}_{1}|\neq0.$
Moreover, consider  the system w.r.t $\tilde{\textbf{z}}\in \mathbb{C}^{N}$ such that $\hbox{supp}(\widehat{\tilde{\textbf{z}}})\subseteq
\{0, \ldots, N/2\}$:
\begin{align}\label{formore12345}\left\{\begin{array}{lllll}
 |\hat{y}_{k,m}|&=&\frac{1}{N}\big|\sum_{l=0}^{k}\hat{\tilde{z}}_{l}\hat{\tilde{z}}_{k-l}w^{lm}\big|, k=1,\cdots, N/2,\\
 |\hat{y}_{k,m}|&=&\frac{1}{N}\big|\sum_{l=k-N/2}^{N/2}\hat{\tilde{z}}_{l}\hat{\tilde{z}}_{k-l}w^{lm}\big|, k= N/2+1,\cdots, N-1,\\
 &&m=0,1,\cdots,r-1, \end{array}\right.  
  \end{align}
where   $\{|\hat{y}_{k,m}|^{2}\}$ are the $[N, L]$-FROG measurements
of $\textbf{z}$ and $r=\lceil N/L\rceil\geq5$.
 Then $\tilde{\textbf{z}}$ is uniquely determined up to  the arbitrary rotation, arbitrary translation  and reflection.
 Moreover, if requiring that the phases of the  first two components of $\widehat{\tilde{\textbf{z}}}$ be zeros,  then $\tilde{\textbf{z}}$
 can be determined (up to the reflection) such that its DFT is
  \begin{align}\label{cvb} \widehat{\tilde{\textbf{z}}}=\big(|\hat{z}_{0}|, |\hat{z}_{1}|, \mathfrak{E}(2)|\hat{z}_{2}|e^{\textbf{i}(\theta_{2}-2\theta_{1})},\ldots,
\mathfrak{E}(N/2)|\hat{z}_{N/2}|e^{\textbf{i}[\theta_{N/2}-(N/2)\theta_{1}]}, 0, \ldots, 0\big),\end{align}
where $\mathfrak{E}(k)$ takes $\hbox{sgn}(\hat{z}_{0})$ and $1$ for $k$ being even and odd, respectively.
\end{theo}
\begin{proof}
Note that the DFT of the  $\gamma$-translation of  $\tilde{\textbf{z}}$
is $(\hat{\tilde{z}}_{0}, \hat{\tilde{z}}_{1}e^{\textbf{i}2\pi\gamma/N},
 \ldots,\hat{\tilde{z}}_{N/2}e^{\textbf{i}2\pi(N/2)\gamma/N}$
 $, 0, \ldots, 0)$. We first prove that
 $(\hat{\tilde{z}}_{0}, \hat{\tilde{z}}_{1}e^{\textbf{i}2\pi\gamma/N},
\ldots,\hat{\tilde{z}}_{N/2}e^{\textbf{i}2\pi(N/2)\gamma/N}$
 $, 0, \ldots, 0)$ satisfies  \eqref{formore12345}.
 By the direct calculation we obtain
%
 \begin{align}\begin{array}{lllll}
   &\frac{1}{N}\big|\sum_{l=0}^{k}\hat{\tilde{z}}_{l}e^{-\textbf{i}2 \pi l\gamma/N}\hat{\tilde{z}}_{k-l}e^{-\textbf{i}2\pi(k-l)\gamma/N}w^{lm}\big|\\
   =&\frac{1}{N}|e^{-\textbf{i}\frac{2k\pi}{N}\gamma}|\big|\sum_{l=0}^{N-1}\hat{\tilde{z}}_{l}\hat{\tilde{z}}_{k-l}w^{lm}\big|\\
   =&|\hat{y}_{k,m}|, m=0,1,\cdots,r-1,
   \end{array} 
  \end{align}
and
 \begin{align}\begin{array}{lllll}
   &\frac{1}{N}\big|\sum_{l=k-N/2}^{N/2}\hat{\tilde{z}}_{l}e^{-\textbf{i}2 \pi l\gamma/N}\hat{\tilde{z}}_{k-l}e^{-\textbf{i}2\pi(k-l)\gamma/N}w^{lm}\big|\\
   =&\frac{1}{N}|e^{-\textbf{i}\frac{2k\pi}{N}\gamma}|\big|\sum_{l=k-N/2}^{N/2}\hat{\tilde{z}}_{l}\hat{\tilde{z}}_{k-l}w^{lm}\big|\\
   =&|\hat{y}_{k,m}|, m=0,1,\cdots,r-1.
   \end{array} 
  \end{align}
 Thus  \eqref{formore12345} is satisfied. By \cite[Proposition 2.2]{Eldar},
the DFT of the  reflection and arbitrary  rotation of  $\tilde{\textbf{z}}$ satisfies  \eqref{formore12345}, respectively.
Since the arbitrary rotation can lead to the ambiguity, we assign  that $\hat{\tilde{z}}_{0}>0$.  By Theorem \ref{auxiliarytheorem} ,  if $|\hat{\tilde{z}}_{0}|\neq |\hat{z}_{0}|$, then $\widehat{\tilde{\textbf{z}}}$  does not satisfy  \eqref{formore12345}.
Now we choose $\hat{\tilde{z}}_{0}=|\hat{z}_{0}|$.  On the other hand,  \begin{equation}\label{forone}
\frac{2}{N}\big|\hat{\tilde{z}}_{0}\hat{\tilde{z}}_{1}\big|=|\hat{y}_{1,0}|=\frac{2}{N}\big|\hat{z}_{0}\hat{z}_{1}\big|,
\end{equation}
from which we derive $|\hat{\tilde{z}}_{1}|=\frac{N|\hat{y}_{1,0}|}{2|\hat{z}_{0}|}=|\hat{z}_{1}|$.
As proved above, any translation of $\tilde{\textbf{z}}$ is also a solution, then we assign  $\hat{\tilde{z}}_{1}$
by $\frac{N|\hat{y}_{1,0}|}{2|\hat{z}_{0}|}$.
Now it follows from Theorem \ref{localtheorem} that
$$\big(|\hat{z}_{0}|, |\hat{z}_{1}|, \mathfrak{E}(2)|\hat{z}_{2}|e^{\textbf{i}(\theta_{2}-2\theta_{1})}, \mathfrak{E}(3)|\hat{z}_{3}|e^{\textbf{i}(\theta_{3}-3\theta_{1})},\ldots,
\mathfrak{E}(N/2)|\hat{z}_{N/2}|e^{\textbf{i}[\theta_{N/2}-(N/2)\theta_{1}]}, 0, \ldots, 0\big)$$
is the solution to \eqref{formore12345}, up to the arbitrary rotation, arbitrary translation  and reflection,
where $\mathfrak{E}(k)$ takes $\hbox{sgn}(\hat{z}_{0})$ and $1$ for $k$ being even and odd, respectively.
Naturally, if requiring the phases of the  first two components of $\widehat{\tilde{\textbf{z}}}$ to be zeros, then
$\tilde{\textbf{z}}$ is the solution  (up to the reflection) to \eqref{formore12345}.
 \end{proof}

\subsection{Determination of generic analytic signals of even lengths by FROG measurements}
\label{section2.5}

In what follows we establish   the  approach for the FROG-PR of  generic analytic signals  of even lengths. This is the key approach which leads to the main result of this paper.

\begin{appr}\label{APPROACH3}
\textbf{Input}: $(3N/2+1)$ FROG  measurements in \eqref{FORGMEASUREMNTS} of $\textbf{z}$. \\
\textbf{Step 1}: Conduct Approach \ref{bvc} to obtain the corresponding output $\hat{\tilde{z}}_{0}$. \\
\textbf{Step 2}: Conduct Approach \ref{approach22} to obtain  the output $\widehat{\tilde{\textbf{z}}}:=(\hat{\tilde{z}}_{0},\hat{\tilde{z}}_{1},\cdots,\hat{\tilde{z}}_{N/2},0,\cdots,0)$.\\
\textbf{Step 3}: Construct  $\widehat{\mathring{\textbf{z}}}=(\hat{\mathring{z}}_{0}, \ldots, \hat{\mathring{z}}_{N/2}, 0, \ldots, 0)\in \mathbb{C}^{N}$,
where  $\hat{\mathring{z}}_{k}=\hat{\tilde{z}}_{k}e^{-\textbf{i}\frac{2 k}{N}\arg(\hat{\tilde{z}}_{N/2})}, k=0, \ldots, N/2$. \\
\textbf{Output}: $\mathring{\textbf{z}}=\hbox{IDFT}(\widehat{\mathring{\textbf{z}}})$. \%\% $\hbox{IDFT}$ is the inverse discrete Fourier transform.
\end{appr}


Now comes to our main theorem.

\begin{theo}\label{twobridges}
Let $\textbf{z}\in \mathbb{C}^{N}$ be a generic   analytic signal such that  $N$ is even and its DFT
\begin{equation}
\widehat{\textbf{z}}=(\hat{z}_{0}, \hat{z}_{1}, \ldots, \hat{z}_{N/2}, 0,
\ldots, 0)
\end{equation}
satisfies  $|\hat{z}_{0}|\neq0 $ and $|\hat{z}_{1}|\neq0.$
Suppose that $\mathring{\textbf{z}}$ is the output of Approach \ref{APPROACH3}
conducted by the $3N/2+1$ $[N, L]$-FROG  measurements in \eqref{FORGMEASUREMNTS} of $\textbf{z}$, where $L$ is odd and $\lceil N/L\rceil\geq 5$.
Then $\mathring{\textbf{z}}$ is the solution to \eqref{FROGFREQ123489} up to the $\pi$-rotation, integer-translation and reflection.
\end{theo}

\begin{proof} 
Note that $\mathring{\textbf{z}}$ is the $(-\frac{\arg(\hat{\tilde{z}}_{N/2})}{\pi}$)-translation of
$\tilde{\textbf{z}}$ in Theorem  \ref{lemma2.6}. Then,  by Theorem  \ref{lemma2.6}, $\mathring{\textbf{z}}$
is a  solution to \eqref{formore12345}, up to the arbitrary rotation, arbitrary translation  and reflection. So  any solution to
\eqref{formore12345} can be derived from the composition of  the  rotation,   translation and  reflection of  $\mathring{\textbf{z}}$.
We next prove that $\mathring{\textbf{z}}$ is a (analytic) solution to \eqref{FROGFREQ123489}.
Note that  both  $\hat{\mathring{z}}_{0}$ and $\hat{\mathring{z}}_{N/2}$ are real-valued. Then,  by Proposition \ref{characterizationeven} (i), $\mathring{\textbf{z}}$ is analytic.
Comparing \eqref{FROGFREQ123489} and \eqref{formore12345}, we only need to check  (\ref{FROGFREQ123489}\hbox{A}). Indeed,
\begin{align}
   \frac{1}{N}\big|\hat{\mathring{z}}^{2}_{0}+\hat{\mathring{z}}^{2}_{N/2}w^{Nm/2}\big|=\frac{1}{N}\big||\hat{z}_{0}|^{2}+|\hat{z}_{N/2}|^{2}w^{Nm/2}\big|=|\hat{y}_{0,m}|,m=0,1,\cdots,r-1.
  \end{align}
Note that any solution to \eqref{FROGFREQ123489} is a solution to \eqref{formore12345}.
Thus  any solution to \eqref{FROGFREQ123489} is the composition of the rotation,   translation and  reflection of $\mathring{\textbf{z}}$, and hence,  by Proposition \ref{remm}, we complete the proof.
\end{proof}





\section{The Proofs}

\subsection{Proof of Proposition \ref{characterizationeven}}\label{appendix1}
We just need to prove (i). Item (ii) can be proved similarly.
By \eqref{up1}, we only need to prove the sufficiency.
Denote $\textbf{z}=\textbf{x}+\textbf{i}\textbf{y}$ such that
where $\widehat{\textbf{x}}$ and $\widehat{\textbf{y}}$,    the DFTs of $\textbf{x}$
and $\textbf{y}$, can be expressed  by
\begin{align}\label{nh1} \hat{x}_{k}=\left\{\begin{array}{lll}
\hat{z}_{0}&k=0, \\
\frac{1}{2}\hat{z}_{k}&1\le k \le N/2-1,\\
\hat{z}_{N/2}&k=N/2,\\
\frac{1}{2}\bar{\hat{z}}_{N-k}&N/2+1 \le k \le N-1,
\end{array}\right.
\end{align}
and
\begin{align}\label{nh2} \textbf{i}\hat{y}_{k}=\left\{\begin{array}{lll}
0,&k=0, \\
\frac{1}{2}\hat{z}_{k},&1\le k \le N/2-1,\\
0,&k=N/2,\\
-\frac{1}{2}\bar{\hat{z}}_{N-k},&N/2+1 \le k \le N-1.
\end{array}\right.
\end{align}
We claim that
$\bar{x}_{k}={x}_{k}$ and $ \bar{y}_{k}={y}_{k}$.
Indeed, it follows from \begin{equation*}
x_{k}= \frac{1}{N}\sum_{n=0}^{N-1}\hat{x}_{n}e^{2\pi \textbf{i}kn/N}, \ y_{k}= \frac{1}{N}\sum_{n=0}^{N-1}\hat{y}_{n}e^{2\pi \textbf{i}kn/N}
\end{equation*}
that
\begin{align}\label{mh1}\begin{array}{lll}
\bar{x}_{k}&=\frac{1}{N}\overline{\sum_{n=0}^{N-1}\hat x_{n}e^{2\pi \textbf{i}kn/N}}\\
&=\frac{1}{N}{\sum_{n=0}^{N-1}\bar{\hat x}_{n}}e^{2\pi \textbf{i}k(N-n)/N}\\
&=\frac{1}{N}{\sum_{n=0}^{N-1}\hat{x}_{N-n}}e^{2\pi \textbf{i}k(N-n)/N}\\
&=\frac{1}{N}{\sum_{n=0}^{N-1}\hat{x}_{n}}e^{2\pi \textbf{i}kn/N}\\
&={x}_{k},
\end{array}
\end{align}
and similarly,
\begin{align}\label{mh2}\begin{array}{lll}
\overline{\textbf{i}y_{k}}&=\frac{1}{N}\overline{\sum_{n=0}^{N-1}\hat y_{n}e^{2\pi \textbf{i}kn/N}}\\
&=\frac{1}{N}{\sum_{n=0}^{N-1}\bar{\hat y}_{n}}e^{2\pi \textbf{i}k(N-n)/N}\\
&=\frac{1}{N}{\sum_{n=0}^{N-1}-\textbf{i}\hat{y}_{N-n}}e^{2\pi \textbf{i}k(N-n)/N}\\
&=\frac{1}{N}{\sum_{n=0}^{N-1}-\textbf{i}\hat{y}_{n}}e^{2\pi ikn/N}\\
&=-\textbf{i}{y}_{k},
\end{array}
\end{align}
where the third identities in \eqref{mh1} and \eqref{mh2}
are derived from \eqref{nh1} and \eqref{nh2}.\qed

\subsection{Proof of Lemma \ref{lemma233}}\label{kkeke}

(1)  If $r$ is odd, then for any $ m\in\{0,1,2,\ldots,r-1\}$, $\frac{4m}{r}$ is not an odd integer and consequently,
$1+w^{2m}=1+e^{\frac{\textbf{i}4m\pi}{r}}\ne0$. Therefore, for the  odd integer $r\ge5$, there exist at least $5$ numbers, denoted by $m_{1}, \ldots, m_{5}\subseteq \{0,1,2,\ldots,r-1\}$ such that $1+w^{2m_{l}}\ne0, l=1, \ldots, 5$. 
 Now assume that  $r$ is even.
 If $r=6$, then $1+w^{2m}=1+e^{\frac{\textbf{i}2m\pi}{3}}$.
 Note that $\frac{2m}{3}$ is not an odd integer. Thus for any $ m\in\{0,1,2,\ldots, 5\}$,
 we have $1+w^{2m}\ne0$.
 If $r>6$ and even (or $r\geq8$), then
 $4m\le4(r-1)< 4r$. If $1+w^{2m}=0$, then $4m=r$ or $4m=3r$. Thus there exist at most two numbers, denoted by $\bar{m}_{1}$
 and $\bar{m}_{2}$, in
 $\{0,1,2,\ldots,r-1\}$ such that $1+w^{2\bar{m}_{k}}=0, k=1, 2$. Therefore there exist at least $5$ numbers, denoted by $m_{1}, \ldots, m_{5}\subseteq \{0,1,2,\ldots,r-1\}$ such that $1+w^{2m_{l}}\ne0 (l=1, \ldots, 5)$.

 (2) Let $m\in \{1,2\ldots,r-1\}$. If
$\frac{w^{m}}{1+w^{2m}}=\frac{1}{2}$, then and so  $(w^{m}-1)^{2}=0$, i.e.,  $w^{m}=1$.  Since  $2\le 2m \le 2(r-1)<2r$, we have  $w^{m}\neq1$, and hence  there exists $m_{1}\in \{1,2,\ldots,r-1\}$ such that $\frac{w^{m_{1}}}{1+w^{2m_{1}}}\ne\frac{1}{2}$.


(3)  We first prove that if  $r\ge5$,  then there exist at least three numbers in $\{1,2,\ldots,r-1\}$, denoted by
 $m_{1}, m_{2}, m_{3}$, such that $1+w^{3m_{k}}=1+e^{\frac{\textbf{i}6m_{k}\pi}{r}}\ne0$.

For  $r=5$ and any $m\in \{1,2,\ldots, 4\}$, $\frac{6m}{r}$ is not an odd integer. Therefore
 $1+w^{3m}=1+e^{\frac{\textbf{i}6m\pi}{5}}\ne0$.
Let $r\geq6$ and $m\in \{1,2,\ldots, r-1\}$.
Suppose that $1+w^{3m}=1+e^{\frac{\textbf{i}6m\pi}{r}}=0$.
Since $\frac{6}{r}\leq\frac{6m}{r}\le\frac{6(r-1)}{r}<6$, we have $1+e^{\frac{\textbf{i}6m\pi}{r}}=0$  which implies that
$\frac{6m}{r}\in \{1, 3, 5\}$. Thus  there exist at most three numbers $\bar{m}_{1},\bar{m}_{2},\bar{m}_{3}\in\{1,2,\ldots,r-1\}$
such that $1+w^{3\bar{m}_{k}}=0$, and therefore there exist at least three numbers $m_{1},m_{2},m_{3}\in\{1,2,\ldots,r-1\}$
such that $1+w^{3m_{k}}\ne0$. %

Now we prove that there exists at least a number $i\in \{m_{1},m_{2},m_{3}\}$
such that $\frac{w^{i}+w^{2i}}{1+w^{3i}}\ne1$.  Note that for any $m\in \{m_{1},m_{2},m_{3}\}$,
$\frac{w^{m}+w^{2m}}{1+w^{3m}}=1$ is equivalent to $(w^{m}-1)(w^{2m}-1)=0$. Thus we have $w^{m}=1$ or $w^{2m}=1$. Since $2\le 2m \le 2(r-1)<2r$, there does not exist $m \in \{m_{1},m_{2},m_{3}\}$ such that $w^{m}=1$.  If $w^{2m}=e^{\frac{\textbf{i}4m\pi}{r}}=1$, then from $4\le4m\le4(r-1)<4r$ we have that that $4m=2r$.  Thus there exists at most one number, e.g.  $m_{1}\in \{m_{1},m_{2},m_{3}\}$ such that  $\frac{w^{m_{1}}+w^{2m_{1}}}{1+w^{3m_{1}}}=1$, and consequently there exists $i\in \{m_{2},m_{3} \}$ such that $\frac{w^{i}+w^{2i}}{1+w^{3i}}\ne1$. \qed

\subsection{Proof of Theorem  \ref{yinli2.1}}\label{proofof34}
Since $\textbf{z}$  is  analytic, we have that both $\hat{z}_{0}$
and $\hat{z}_{N/2}$ are real-valued.
We first  consider the   equation system w.r.t the real-valued  variables   $\hat{\tilde{z}}_{0}, \hat{\tilde{z}}_{N/2}\in \mathbb{R}$:
\begin{align}\label{pp4}  \left\{
\begin{aligned}
\frac{1}{N}\big|\hat{z}_{0}^{2}+\hat{z}_{N/2}^{2}\big|=|\hat{y}_{0,0}|&=\frac{1}{N}\big|\hat{\tilde{z}}_{0}^{2}+\hat{\tilde{z}}_{N/2}^{2}\big|, \quad (\ref{pp4} A) \\
\frac{1}{N}\big|\hat{z}_{0}^{2}-\hat{z}_{N/2}^{2}\big|=|\hat{y}_{0,1}|&=\frac{1}{N}\big|\hat{\tilde{z}}_{0}^{2}-\hat{\tilde{z}}_{N/2}^{2}\big|. \quad (\ref{pp4} B)
\end{aligned}
\right.
\end{align}
It is easy to check that  the solutions to \eqref{pp4} are   \begin{equation}
(\hat{\tilde{z}}_{0}, \hat{\tilde{z}}_{N/2})=\Big(\pm\sqrt{\frac{N(|\hat{y}_{0,0}|+|\hat{y}_{0,1}|)}{2}}, \pm\sqrt{\frac{N(|\hat{y}_{0,0}|-|\hat{y}_{0,1}|)}{2}}\Big)
\end{equation}
and
\begin{equation}
(\hat{\tilde{z}}_{0}, \hat{\tilde{z}}_{N/2})=\Big(\pm\sqrt{\frac{N(|\hat{y}_{0,0}|-|\hat{y}_{0,1}|)}{2}}, \pm\sqrt{\frac{N(|\hat{y}_{0,0}|+|\hat{y}_{0,1}|)}{2}}\Big).
\end{equation}
If $|\hat{y}_{0,1}|=0$,  then $|\hat{\tilde{z}}_{0}|=|\hat{\tilde{z}}_{N/2}|$ and consequently $\hat{z}_{0}$ can be determined (up to a sign) by Approach \ref{bvc}

If   $|\hat{y}_{0,1}|\neq0$ (equivalently  $|\hat{z}_{0}|\neq|\hat{z}_{N/2}|$), then we need to prove that
 if   a signal  $\tilde{\textbf{z}}\in \mathbb{C}^{N}$ whose DFT $(\hat{\tilde{z}}_{0}, \ldots, \hat{\tilde{z}}_{N/2}, 0, \ldots, 0)$ satisfies  \begin{align}\label{ccf1} \hat{\tilde{z}}_{0}\in \{\hat{z}_{N/2}, -\hat{z}_{N/2}\} \ \hbox{and} \ \hat{\tilde{z}}_{N/2}\in \{\hat{z}_{0}, -\hat{z}_{0}\},\end{align}
 then  it  does not have the same $[N, L]$-FROG measurements  as $\textbf{z}$. This follows from Theorem \ref{auxiliarytheorem}. \qed

\subsection{Proof of Theorem \ref{localtheorem}}\label{proofoftheorem2.5}

We divide the proof into three parts based on the three steps from Approach \ref{approach22}


\subsubsection{Determination of $\hat{\tilde{z}}_{2}$ in \eqref{localequation}}\label{zqzill}
\begin{prop}\label{lemma31}
Suppose  that   $N$ is even and $\textbf{z}\in \mathbb{C}^{N}$  is  a generic  analytic signal
such that  its DFT \begin{align} \label{ccfur} \widehat{\textbf{z}}=(\hat{z}_{0}, \hat{z}_{1}, \ldots, \hat{z}_{N/2},0,
\ldots, 0):=(|\hat{z}_{0}|e^{\textbf{i}\theta_{0}}, |\hat{z}_{1}|e^{\textbf{i}\theta_{1}}, \ldots, |\hat{z}_{N/2}|e^{\textbf{i}\theta_{N/2}},
0,\ldots, 0)\end{align} satisfies $|\hat{z}_{0}|, |\hat{z}_{1}|\neq0.$
The $[N, L]$-FROG measurements of $\textbf{z}$ are $\{|\hat{y}_{k,m}|^{2}: k=0,1,\cdots,N-1;  m=0, \ldots, r-1\}$.
Denote by $\hat{\tilde{z}}_{0}$ the output of Approach \ref{bvc}
 such that $\hat{\tilde{z}}_{0}=\epsilon\hat{z}_{0}$ with $\epsilon\in \{1, -1\}$.
Define $\hat{\tilde{z}}_{1}:=\frac{N|\hat{y}_{1,0}|}{2|\hat{z}_{0}|}$ such that $
\frac{2}{N}\big|\hat{\tilde{z}}_{0}\hat{\tilde{z}}_{1}\big|=\frac{2}{N}\big|\hat{z}_{0}\hat{z}_{1}\big|=|\hat{y}_{1,0}|
$.
Then the solutions to the  system of equations w.r.t $\hat{\tilde{z}}_{2}$:
\begin{align}\label{check}
\frac{1}{N}\big|\hat{\tilde{z}}_{0}\hat{\tilde{z}}_{2}(1+w^{2m})+\hat{\tilde{z}}_{1}^{2}w^{m}\big|=
|\hat{y}_{2,m}|, m=0, \ldots, r-1,
\end{align}
are $\epsilon|\hat{z}_{2}|e^{\textbf{i}(\theta_{2}-2\theta_{1})}$ and $\epsilon|\hat{z}_{2}|e^{\textbf{i}(-\theta_{2}+2\theta_{1})}$. Moreover, one of the two solutions can be determined by \textbf{Step 1} in Approach \ref{approach22}.
\end{prop}
\begin{proof}
Clearly,
$|\hat{\tilde{z}}_{1}|=|\hat{z}_{1}|$ and  $\hat{\tilde{z}}_{1}=|\hat{z}_{1}|e^{\textbf{i}(\theta_{1}-\theta_{1})}$.
For $\hat{\tilde{z}}_{2}=\epsilon|\hat{z}_{2}|e^{\textbf{i}(\theta_{2}-2\theta_{1})}$ and $m=0,1,\ldots,r-1$, compute
\begin{align}\label{ytr}\begin{array}{llll}
\frac{1}{N}\big|\hat{\tilde{z}}_{0}\hat{\tilde{z}}_{2}+\hat{\tilde{z}}_{1}^{2}w^{m}+\hat{\tilde{z}}_{2}\hat{\tilde{z}}_{0}w^{2m}\big| \\
=\frac{1}{N}\big|\hat{z}_{0}|\hat{z}_{2}|e^{\textbf{i}(\theta_{2}-2\theta_{1})}+(|\hat{z}_{1}|e^{\textbf{i}(\theta_{1}-\theta_{1})})^{2}w^{m}+|\hat{z}_{2}|e^{\textbf{i}(\theta_{2}-2\theta_{1})}\hat{z}_{0}w^{2m}\big|\\
=\frac{1}{N}|e^{\textbf{i}(2\theta_{1})}|\big|\hat{z}_{0}\hat{z}_{2}+\hat{z}_{1}^{2}w^{m}+\hat{z}_{2}\hat{z}_{0}w^{2m}\big| \\
=\frac{1}{N}\big|\hat{z}_{0}\hat{z}_{2}+\hat{z}_{1}^{2}w^{m}+\hat{z}_{2}\hat{z}_{0}w^{2m}\big| \\
=|\hat{y}_{2,m}|.
\end{array}\end{align}
Therefore,  $\hat{\tilde{z}}_{2}=\epsilon|\hat{z}_{2}|e^{\textbf{i}(\theta_{2}-2\theta_{1})}$
is the solution to \eqref{check}. Moreover,
it  follows from Lemma \ref{lemma233} (2) that there exists $m_{1}\in\{ 1, 2, \ldots, r-1\}$ such that
$\frac{w^{m_{1}}}{1+w^{2m_{1}}}\ne\frac{1}{2}$. For $m=0$ and $ m_{1}$,
\eqref{check} is equivalent to
\begin{align}\label{treg}
\frac{|\hat{y}_{2,m}|}{|(1+w^{2m})\hat{\tilde{z}}_{0}|}=\frac{1}{N}\big|\hat{\tilde{z}}_{2}+\frac{w^{m}}{1+w^{2m}}\frac{\hat{\tilde{z}}_{1}^{2}}{\hat{\tilde{z}}_{0}}\big|.
\end{align}
It is easy to check that $\frac{w^{m}}{1+w^{2m}}\in \mathbb{R}$ for any $m\in \{0, m_{1}\}$. Thus,  by Lemma \ref{lemmm2.2} (i) and \eqref{ytr},  the solutions to \eqref{treg} are $\epsilon|\hat{z}_{2}|e^{\textbf{i}(\theta_{2}-2\theta_{1})}$ and $\epsilon|\hat{z}_{2}|e^{\textbf{i}(-\theta_{2}+2\theta_{1})}$. By the similar calculation as \eqref{ytr}, $\epsilon|\hat{z}_{2}|e^{\textbf{i}(-\theta_{2}+2\theta_{1})}$ is the other solution to \eqref{check}.
The proof is concluded.
\end{proof}

\subsubsection{Determination of $\hat{\tilde{z}}_{3},\hat{\tilde{z}}_{4}$  in \eqref{localequation}}\label{zqzi22}
Let  $\textbf{z}\in \mathbb{C}^{N}$
be a generic  analytic signal  such that its $[N, L]$-FROG measurements   are $\{|\hat{y}_{k,m}|^{2}: k=0,1,\cdots,N-1;  m=0, \ldots, r-1\}$. Denote  its DFT
by
\begin{align} \label{ccfurrt} \widehat{\textbf{z}}=(\hat{z}_{0}, \hat{z}_{1}, \ldots, \hat{z}_{N/2},0,
\ldots, 0):=(|\hat{z}_{0}|e^{\textbf{i}\theta_{0}}, |\hat{z}_{1}|e^{\textbf{i}\theta_{1}}, \ldots, |\hat{z}_{N/2}|e^{\textbf{i}\theta_{N/2}},0,
\ldots, 0).\end{align}
Recall that in \eqref{localequation}, $\hat{\tilde{z}}_{0}=\epsilon \hat{z}_{0}$ with $ \epsilon\in \{1, -1\}$ is  the output of Approach \ref{bvc}, and    $\hat{\tilde{z}}_{1}=|\hat{z}_{1}|$.
 Consequently, by the similar analysis in the proof of Proposition \ref{lemma31}, $\epsilon|\hat{z}_{2}|e^{\textbf{i}(\theta_{2}-2\theta_{1})}$ and $\epsilon|\hat{z}_{2}|e^{\textbf{i}(-\theta_{2}+2\theta_{1})}$
 are the solutions to \eqref{check} w.r.t $\hat{\tilde{z}}_{2}$.

 \begin{prop}\label{yinli2.5}
Let the generic analytic signal $\textbf{z}\in \mathbb{C}^{N}$ be as above.
 If we choose  $\hat{\tilde{z}}_{2}=\epsilon|\hat{z}_{2}|e^{\textbf{i}(\theta_{2}-2\theta_{1})}$,  then the solution to the  system of equations w.r.t $\hat{\tilde{z}}_{3}$ and $ \hat{\tilde{z}}_{4}$:
\begin{align}\label{lemma54}\left\{\begin{array}{lll}
|\hat{y}_{3,m}|&=&\frac{1}{N}\big|\hat{\tilde{z}}_{0}\hat{\tilde{z}}_{3}+\hat{\tilde{z}}_{1}\hat{\tilde{z}}_{2}w^{m}+\hat{\tilde{z}}_{2}\hat{\tilde{z}}_{1}w^{2m}
+\hat{\tilde{z}}_{3}\hat{\tilde{z}}_{0}w^{3m}\big|,\\
|\hat{y}_{4,m}|&=&\frac{1}{N}\big|\hat{\tilde{z}}_{0}\hat{\tilde{z}}_{4}+\hat{\tilde{z}}_{1}\hat{\tilde{z}}_{3}w^{m}+(\hat{\tilde{z}}_{2})^{2}w^{2m}
+\hat{\tilde{z}}_{3}\hat{\tilde{z}}_{1}w^{3m}+\hat{\tilde{z}}_{4}\hat{\tilde{z}}_{0}w^{4m}\big|,\\
&&m=0, 1, \ldots, r-1,
\end{array}\right.\end{align}
is $(|\hat{z}_{3}|e^{\textbf{i}(\theta_{3}-3\theta_{1})}, \epsilon|\hat{z}_{4}|e^{\textbf{i}(\theta_{4}-4\theta_{1})})$. Similarly, if
$\hat{\tilde{z}}_{2}=\epsilon|\hat{z}_{2}|e^{\textbf{i}(-\theta_{2}+2\theta_{1})}$, then the solution to \eqref{lemma54}
is $(|\hat{z}_{3}|e^{\textbf{i}(-\theta_{3}+3\theta_{1})}, \epsilon|\hat{z}_{4}|e^{\textbf{i}(-\theta_{4}+4\theta_{1})})$. Moreover, the solution can be determined by \textbf{Step 2} in Approach \ref{approach22}.
\end{prop}

%

\begin{proof}
Clearly,    $\hat{\tilde{z}}_{1}=|\hat{z}_{1}|e^{\textbf{i}(\theta_{1}-\theta_{1})}$.
Suppose that we choose  $\hat{\tilde{z}}_{2}=\epsilon|\hat{z}_{2}|e^{\textbf{i}(\theta_{2}-2\theta_{1})}$. Then it follows from
\begin{align*}\begin{array}{llll}
&\frac{1}{N}\big|\hat{\tilde{z}}_{0}\hat{\tilde{z}}_{3}+\hat{\tilde{z}}_{1}\hat{\tilde{z}}_{2}w^{m}+\hat{\tilde{z}}_{2}\hat{\tilde{z}}_{1}w^{2m}+
\hat{\tilde{z}}_{3}\hat{\tilde{z}}_{0}w^{3m}\big| \\
=& \ \frac{1}{N}\big|\epsilon\hat{z}_{0}|\hat{z}_{3}|e^{\textbf{i}(\theta_{3}-3\theta_{1})}+|\hat{z}_{1}|e^{\textbf{i}(\theta_{1}-\theta_{1})}\epsilon|\hat{z}_{2}|e^{\textbf{i}(\theta_{2}-2\theta_{1})}w^{m}+
\epsilon|\hat{z}_{2}|e^{\textbf{i}(\theta_{2}-2\theta_{1})}|\hat{z}_{1}|e^{\textbf{i}(\theta_{1}-\theta_{1})}w^{2m}\\
&\quad +|\hat{z}_{3}|e^{\textbf{i}(\theta_{3}-3\theta_{1})}\epsilon\hat{z}_{0}w^{3m}\big|\\
=&\frac{1}{N}|e^{\textbf{i}3(-\theta_{1})}|\big|\hat{z}_{0}\hat{z}_{3}+\hat{z}_{1}\hat{z}_{2}w^{m}+\hat{z}_{2}\hat{z}_{1}w^{2m}+\hat{z}_{0}\hat{z}_{3}w^{3m}\big| \\
=&\frac{1}{N}\big|\hat{z}_{0}\hat{z}_{3}+\hat{z}_{1}\hat{z}_{2}w^{m}+\hat{z}_{2}\hat{z}_{1}w^{2m}+\hat{z}_{0}\hat{z}_{3}w^{3m}\big| \\
=&|\hat{y}_{3,m}|\end{array}
\end{align*}
that
$\hat{\tilde{z}}_{3}=|\hat{z}_{3}|e^{\textbf{i}(\theta_{3}-3\theta_{1})}$ is a solution to
\begin{align} \label{foggy}|\hat{y}_{3,m}|=\frac{1}{N}\big|\hat{\tilde{z}}_{0}\hat{\tilde{z}}_{3}+\hat{\tilde{z}}_{1}\hat{\tilde{z}}_{2}w^{m}+\hat{\tilde{z}}_{2}\hat{\tilde{z}}_{1}w^{2m}
+\hat{\tilde{z}}_{3}\hat{\tilde{z}}_{0}w^{3m}\big|, m=0,1,\ldots,r-1.\end{align}
As for the above system, it follows from the
analysis  in \cite[Page 1037]{Eldar}
that $\hat{\tilde{z}}_{3}'=|\hat{z}_{3}|e^{\textbf{i}(-\theta_{3}+3\theta_{1}+2\tilde{\theta}_{2})}$
with $\tilde{\theta}_{2}=\arg(\hat{\tilde{z}}_{2})$ is the other solution to \eqref{foggy}.
By \cite[Lemma 4.2]{Eldar}, however,   if we choose $\hat{\tilde{z}}_{3}'$ as the  solution, then there does not exist any other solution to \begin{align} \label{uytr5} |\hat{y}_{4,m}|=\frac{1}{N}\big|\hat{\tilde{z}}_{0}\hat{\tilde{z}}_{4}+\hat{\tilde{z}}_{1}\hat{\tilde{z}}_{3}w^{m}+(\hat{\tilde{z}}_{2})^{2}w^{2m}
+\hat{\tilde{z}}_{3}\hat{\tilde{z}}_{1}w^{3m}+\hat{\tilde{z}}_{4}\hat{\tilde{z}}_{0}w^{4m}\big|,m=0,1,\ldots,r-1. \end{align}
Therefore, we just need to consider   $\hat{\tilde{z}}_{3}=|\hat{z}_{3}|e^{\textbf{i}(\theta_{3}-3\theta_{1})}$.

Having  $\hat{\tilde{z}}_{2}=\epsilon|\hat{z}_{2}|e^{\textbf{i}(\theta_{2}-2\theta_{1})}$ and $\hat{\tilde{z}}_{3}=|\hat{z}_{3}|e^{\textbf{i}(\theta_{3}-3\theta_{1})}$, we can verify that
$\hat{\tilde{z}}_{4}=\epsilon|\hat{z}_{4}|e^{\textbf{i}(\theta_{4}-4\theta_{1})}$ is a solution to
\eqref{uytr5}.
Indeed this follows from:
\begin{align*}\begin{array}{lllll}
&\frac{1}{N}\Big|\hat{\tilde{z}}_{0}\hat{\tilde{z}}_{4}+\hat{\tilde{z}}_{1}\hat{\tilde{z}}_{3}w^{m}+(\hat{\tilde{z}}_{2})^{2}w^{2m}+\hat{\tilde{z}}_{3}\hat{\tilde{z}}_{1}w^{3m}
+\hat{\tilde{z}}_{4}\hat{\tilde{z}}_{0}w^{4m}\Big| \\
=&\ \frac{1}{N}\big|\epsilon^{2}\hat{z}_{0}|\hat{z}_{4}|e^{\textbf{i}(\theta_{4}-4\theta_{1})}+|\hat{z}_{1}|e^{\textbf{i}(\theta_{1}-\theta_{1})}|\hat{z}_{3}|e^{\textbf{i}(\theta_{3}-3\theta_{1})}w^{m}
+\epsilon^{2}[|\hat{z}_{2}|e^{\textbf{i}(\theta_{2}-2\theta_{1})}]^{2}w^{2m}\\
&\quad+|\hat{z}_{3}|e^{\textbf{i}(\theta_{3}-3\theta_{1})}|\hat{z}_{1}|e^{\textbf{i}(\theta_{1}-\theta_{1})}w^{3m}+
\epsilon^{2}|\hat{z}_{4}|e^{\textbf{i}(\theta_{4}-4\theta_{1})}\hat{z}_{0}w^{4m}\big|\\
=&\frac{1}{N}|e^{\textbf{i}4(-\theta_{1})}|\big|\hat{z}_{0}\hat{z}_{4}+\hat{z}_{1}\hat{z}_{3}w^{m}+\hat{z}_{2}^{2}w^{2m}+\hat{z}_{3}\hat{z}_{1}w^{3m}+\hat{z}_{0}\hat{z}_{4}w^{4m}\big| \\
=&\frac{1}{N}\big|\hat{z}_{0}\hat{z}_{4}+\hat{z}_{1}\hat{z}_{3}w^{m}+\hat{z}_{2}^{2}w^{2m}+\hat{z}_{3}\hat{z}_{1}w^{3m}+\hat{z}_{0}\hat{z}_{4}w^{4m}\big| \\
=&|\hat{y}_{4,m}|.\end{array}
\end{align*}

We next prove that $\hat{\tilde{z}}_{4}=\epsilon|\hat{z}_{4}|e^{\textbf{i}(\theta_{4}-4\theta_{1})}$
is the unique solution to \eqref{uytr5} for the generic analytic signal $\textbf{z}$.
Suppose that $\hat{\tilde{z}}_{0}, \hat{\tilde{z}}_{1}$ and $
\hat{\tilde{z}}_{2}$ are fixed.
Since   $\hat{z}_{3}$ is generic then   $\hat{\tilde{z}}_{3}=|\hat{z}_{3}|e^{\textbf{i}(\theta_{3}-3\theta_{1})}=\hat{z}_{3}e^{\textbf{i}3\theta_{1}}$ is
also  generic.
We next use Lemma \ref{lemmm2.1} to determine $\hat{\tilde{z}}_{4}$. For the generic
$\hat{\tilde{z}}_{3}$, recall that the corresponding  $\frac{v_{1}-v_{2}}{v_{1}-v_{3}}$
in Lemma \ref{lemmm2.1} \eqref{condition} is a rational polynomial  w.r.t $\hat{\tilde{z}}_{3}$.
Therefore  it is easy to check that \eqref{condition}  holds, and
the solution to \eqref{uytr5} is unique
for the choice of $(\hat{\tilde{z}}_{0}, \hat{\tilde{z}}_{1},
\hat{\tilde{z}}_{2}, \hat{\tilde{z}}_{3})$
$=(\epsilon\hat{z}_{0}, |\hat{z}_{1}|, \epsilon|\hat{z}_{2}|e^{\textbf{i}(\theta_{2}-2\theta_{1})}, $
$ |\hat{z}_{3}|e^{\textbf{i}(\theta_{3}-3\theta_{1})})$.
And  $(\hat{\tilde{z}}_{3},\hat{\tilde{z}}_{4})$ can be determined by \textbf{Step 2}  in Approach \ref{approach22}.
 Now we concluded the proof of the first part, and  the second part can be proved similarly.
\end{proof}


\subsubsection{Determination of $\hat{\tilde{z}}_{k}$  in \eqref{localequation} for $k\geq5$}\label{zqzi33}
We first prove that  if $(\hat{\tilde{z}}_{0}, \cdots ,\hat{\tilde{z}}_{k-1})=\big(\epsilon\hat{z}_{0}, |\hat{z}_{1}|, \cdots ,
 \mathfrak{E}(k-1)|\hat{z}_{k-1}|e^{\textbf{i}[\theta_{k-1}-(k-1)\theta_{1}]}\big)$, then
$\hat{\tilde{z}}_{k}=\mathfrak{E}(k)|\hat{z}_{k}|e^{\textbf{i}(\theta_{k}-k\theta_{1})}$ satisfies

 \begin{align}\label{zk}
  |\hat{y}_{k,m}|=\frac{1}{N}\big|\sum_{l=0}^{k}\hat{\tilde{z}}_{l}\hat{\tilde{z}}_{k-l}w^{lm}\big|,m=0,1,\ldots,r-1.
 \end{align}

Indeed,  from \begin{align}\begin{array}{lllll}
&\frac{1}{N}\big|\sum_{l=0}^{k}\hat{\tilde{z}}_{l}\hat{\tilde{z}}_{k-l}w^{lm}\big|\\
=&\frac{1}{N}\big|\sum_{l=0}^{k}|\hat{z}_{l}|e^{\textbf{i}(\theta_{l}-l\theta_{1})}|\hat{z}_{k-l}|e^{\textbf{i}[\theta_{k-l}-(k-l)\theta_{1}]}w^{lm}\big|\\
=&\frac{1}{N}|e^{\textbf{i}(k\alpha)}|\big|\sum_{l=0}^{k}\hat{z}_{l}\hat{z}_{k-l}w^{lm}\big| \\
=&|\hat{y}_{k,m}|,m=0,1,\ldots,r-1.\end{array}
\end{align}
we get that $\hat{\tilde{z}}_{k}=\mathfrak{E}(k)|\hat{z}_{k}|e^{\textbf{i}(\theta_{k}-k\theta_{1})}$ satisfies \eqref{zk}.
Through the similar  analysis in  the proof of Proposition \ref{yinli2.5} that
$\hat{\tilde{z}}_{k}=\mathfrak{E}(k)|\hat{z}_{k}|e^{\textbf{i}(\theta_{k}-k\theta_{1})}$ is the unique solution to  \eqref{zk}. Similarly,
we can prove that if $(\hat{\tilde{z}}_{0}, \cdots ,\hat{\tilde{z}}_{k-1})=\big(\epsilon\hat{z}_{0}, |\hat{z}_{1}|, \cdots ,
 \mathfrak{E}(k-1)|\hat{z}_{k-1}|e^{\textbf{i}[-\theta_{k-1}+(k-1)\theta_{1}]}\big)$, then
$\hat{\tilde{z}}_{k}=\mathfrak{E}(k)|\hat{z}_{k}|e^{\textbf{i}(-\theta_{k}+k\theta_{1})}$ is the unique solution to \eqref{zk}.
Now by Lemma \ref{lemmm2.1},  $\hat{\tilde{z}}_{k}$ can be determined through \textbf{Step 3} in  Approach \ref{approach22}.

\subsubsection{ \eqref{localequation} has two complex conjugation solutions}
According to sections \ref{zqzill}, \ref{zqzi22} and \ref{zqzi33}, the solutions to  \eqref{localequation}
are the vector  $\big(\mathfrak{E}(2)|\hat{z}_{2}|e^{\textbf{i}(\theta_{2}-2\theta_{1})}, \mathfrak{E}(3)|\hat{z}_{3}|e^{\textbf{i}(\theta_{3}-3\theta_{1})},\ldots,$
$\mathfrak{E}(N/2)|\hat{z}_{N/2}|e^{\textbf{i}[\theta_{N/2}-(N/2)\theta_{1}]}\big)$  and
its complex conjugate.

\section{A final remark}\label{future}
If $L$ is even and $r=\lceil N/L\rceil\geq5$, then
\begin{equation}\label{eq1891}
\frac{1}{N}(\hat{\tilde{z}}_{0}^{2}+\hat{\tilde{z}}_{N/2}^{2})=|\hat{y}_{0,m}|, m=0,1,\cdots,r-1.
\end{equation}
Therefore in this case  $|\hat{y}_{0,m}|$ can not be expressed by  the form of (\ref{pp4}B).
Clearly,    $\hat{\tilde{z}}_{0}\in\big[-\sqrt{N|\hat{y}_{0,0}|},\sqrt{N|\hat{y}_{0,0}|}\big]$.
Assume that we assign $\hat{\tilde{z}}_{0}$  by  an arbitrary value $\alpha\in \big[-\sqrt{N|\hat{y}_{0,0}|},\sqrt{N|\hat{y}_{0,0}|}\big]$.
Then it holds that  $\alpha\neq \pm\hat{z}_{0}$ with probability $1$.
Now it follows from $\frac{2}{N}\big|\hat{\tilde{z}}_{0}\hat{\tilde{z}}_{1}\big|=|\hat{y}_{1,0}|$ that
$|\hat{\tilde{z}}_{1}|=\frac{N|\hat{y}_{1,0}|}{2|\alpha|}$. For any $\theta\in [0, 2\pi)$, assign the phase $\theta$  to  $\hat{\tilde{z}}_{1}$ as $\frac{N|\hat{y}_{1,0}|}{2|\alpha|}e^{\textbf{i}\theta}$. Then,  by the proof of Theorem \ref{auxiliarytheorem}, the system w.r.t $\hat{\tilde{z}}_{2}$
 \begin{align}\label{eq2}
\frac{1}{N}\big|\hat{\tilde{z}}_{0}\hat{\tilde{z}}_{2}(1+w^{2m})+\hat{\tilde{z}}_{1}^{2}w^{m}\big|=
|\hat{y}_{2,m}|, m=0, \ldots, r-1
\end{align}
does not have a solution. This implies that,  with probability $1$,  Approach \ref{APPROACH3}
does not hold  for the FROG-PR problem  with
even $L$  and $r=\lceil N/L\rceil\geq5$, and  so a different approach need to be developed for this case in the future work.

\end{document}